\documentclass[twoside,journal]{IEEEtran}

\usepackage{amsfonts}
\usepackage{amssymb}
\usepackage{amsmath}
\usepackage{mathrsfs}
\usepackage{verbatim}
\usepackage{subfigure}
\usepackage{balance}
\usepackage{booktabs}
\usepackage{fancybox}
\usepackage{bm}
\usepackage{changepage}
\usepackage{extarrows}
\usepackage{algorithm}
\usepackage{algorithmic}
\usepackage{diagbox}
\usepackage{multirow}
\usepackage{array}
\usepackage{graphicx}
\usepackage{epstopdf}
\usepackage{caption}
\usepackage{color}
\newtheorem{theorem}{Theorem}
\newtheorem{lemma}{Lemma}
\newtheorem{corollary}{Corollary}
\newtheorem{proof}{Proof}

\begin{document}

\title{Physical Layer Security in Wireless Ad Hoc Networks Under A Hybrid Full-/Half-Duplex Receiver Deployment Strategy}
\author{Tong-Xing~Zheng,~\IEEEmembership{Member,~IEEE,}
~Hui-Ming~Wang,~\IEEEmembership{Senior Member,~IEEE,~}\\
Jinhong~Yuan,~\IEEEmembership{Fellow,~IEEE,~}~Zhu~Han,~\IEEEmembership{Fellow,~IEEE,~}and~Moon~Ho~Lee,~\IEEEmembership{Life~Senior~Member,~IEEE}
\thanks{T.-X.~Zheng and H.-M. Wang are with the School of Electronic and Information Engineering,
		Xi'an Jiaotong University, Xi'an, 710049, Shaanxi, China. Email: {\tt txzheng@stu.xjtu.edu.cn},
		{\tt xjbswhm@gmail.com}.}
	\thanks{J.~Yuan is with the School of Electrical Engineering and Telecommunications, University of New South Wales, Sydney, Australia. Email: {\tt j.yuan@unsw.edu.au}.}
	\thanks{Z.~Han is with the Electrical and Computer Engineering Department, University of Houston, Houston, TX, USA. Email: {\tt zhan2@uh.edu}.}
	\thanks{M. H. Lee is with the Division of Electronics Engineering, Chonbuk National 	University, Jeonju 561-756, Korea. Email: {\tt moonho@jbnu.ac.kr}.}
}

\maketitle
\vspace{-0.8 cm}

\begin{abstract}
This paper studies physical layer security in a wireless ad hoc network with numerous legitimate transmitter-receiver pairs and eavesdroppers.
A hybrid full-/half-duplex receiver deployment strategy is proposed to secure legitimate transmissions, by letting a fraction of legitimate receivers work in the full-duplex (FD) mode sending jamming signals to confuse eavesdroppers upon their information receptions, and letting the other receivers work in the half-duplex mode just receiving their desired signals.
The objective of this paper is to choose properly the fraction of FD receivers for achieving the optimal network security performance.
Both accurate expressions and tractable approximations for the connection outage probability and the secrecy outage probability of an arbitrary legitimate link are derived, based on which the area secure link number, network-wide secrecy throughput and network-wide secrecy energy efficiency are optimized respectively.
Various insights into the optimal fraction are further developed and its closed-form expressions are also derived under perfect self-interference cancellation or in a dense network.
It is concluded that the fraction of FD receivers triggers a non-trivial trade-off between reliability and secrecy, and the proposed strategy can significantly enhance the network security performance.

\end{abstract}


\IEEEpeerreviewmaketitle

\section{Introduction}

\IEEEPARstart{T}{he} rapid development in wireless communications has brought unprecedented attention to information security.
Traditionally, security issues are addressed at the upper layers of communication protocols by using encryption.
However, the large-scale and dynamic topologies in emerging wireless networks pose a great challenge in implementing secret key management and distribution, particularly in a decentralized wireless ad hoc network without infrastructure support \cite{Poor2012Information}.
Fortunately, \emph{physical layer security}, an information-theoretic approach that attains secure transmissions by exploiting the randomness of wireless channels without necessarily relying on secret keys, is becoming increasingly recognized as a promising alternative to complement the cryptography-based security mechanisms \cite{Wyner1975Wire-tap}-\cite{Goel2008Guaranteeing}.
%In the past decade, physical layer security has been extensively applied for achieving security in broadcast channels \cite{Geraci2013Large,Liu2010Multiple}, multi-antenna channels \cite{Liu2010Multiple}-\cite{Zheng2015Multi}, and cooperative relay channels \cite{Dong2010Improving}-\cite{Wang2012Distributed}, etc.
%A more thorough reference work in the fundamental theory and the evolution of physical layer security can be
%found in \cite{Mukherjee2014Principles}.

Early studies on physical layer security have mainly focused on point-to-point transmissions, and metrics from user viewpoint such as secrecy capacity \cite{Liu2010Multiple}, ergodic secrecy rate \cite{Zhou2010Secure,Wang2015Secure} and secrecy outage probability \cite{Zheng2015Outage}-\cite{Zheng2016Optimal} have been used to evaluate the secrecy level in different scenarios/applications.
	From a \emph{network-wide} perspective, physical layer security has also shown its potential  \cite{Wang2016Physical,Wang2016Physical Springer}.
	Many efforts have already been devoted to improve network security in terms of the area secure link number (ASLN) \cite{Ma2015Interference,Wang Chao2016Physical} and network-wide secrecy throughput (NST) \cite{Zhou2011Throughput,Zhang2013Enhancing}.
	More recently, energy-efficient green wireless network has attracted considerable interests due to energy scarcity, and some research works have been carried out for enhancing network-wide secrecy energy efficiency (NSEE) \cite{Ng2012Energy,Chen2013Energy}.

\subsection{Previous Endeavors and Motivations}
To improve the secrecy of information delivery for an ad hoc network, an efficient approach is to degrade the wiretapping ability of eavesdroppers through emitting jamming signals \cite{Goel2008Guaranteeing}.
For example, the authors in \cite{Zhou2011Throughput}  propose a cooperative jamming strategy with single-antenna legitimate transmitters, and when eavesdroppers access a transmitter's secrecy guard zone \cite{Zhou2011Throughput}, this transmitter will act as a friendly jammer to send jamming signals to confuse eavesdroppers.
This work is extended by \cite{Zhang2013Enhancing} to a multi-antenna transmitter scenario, and artificial noise \cite{Goel2008Guaranteeing} with either sectoring or beamforming  is exploited to impair eavesdroppers.
Although these endeavors are shown to achieve a remarkable secrecy throughput enhancement, friendly jammers or multi-antenna transmitters might not be available in many applications.
For instance, constrained by the size and hardware cost, a sensor node is usually equipped with only a single antenna.
Furthermore, due to a low-power constraint, a sensor has no
extra power to send jamming signals.
In such unfavorable situations, information transfer is still vulnerable to eavesdropping.

Fortunately, recent advances in developing in-band full-duplex (FD) radios provide a new opportunity to strengthen information security in the aforementioned situations. Effective self-interference cancellation (SIC) techniques enable a transceiver to transmit and receive at the same time on the same frequency band \cite{Song2016Full}.
Although the transmitter (sensor) is
vulnerable to eavesdropping, we can deploy powerful FD receivers such as data collection stations to radiate jamming signals upon their information receptions.
By doing so, additional degrees of freedom can be gained for improving network security.
In fact, the idea of using FD receiver jamming to improve physical layer security has already been reported by \cite{Li2012Secure}-\cite{Parsaeefard2015Improving} for point-to-point transmission scenarios.
Specifically, the authors in  \cite{Li2012Secure} and \cite{Zheng2013Improving} consider a single-input multi-output (SIMO) channel with the receiver using single- and multi-antenna jamming, respectively.
The authors in \cite{Zhou2014Application}
consider a multi-input multi-output (MIMO) channel with both transmitter and receiver generating artificial noise.
These works are further extended in two-way  transmissions \cite{Cepheli2014A_high}, cooperative communications
\cite{Chen2015Physical},
\cite{Parsaeefard2015Improving}, and cellular networks \cite{Zhu2016Physical}.
Recently, we have studied the design of the optimal density of the overlaid FD-mode tier to maximize its NST while guaranteeing a minimum network-wide throughput for the underlaid HD-mode tier \cite{Zheng2017Safeguarding}.
Generally, investigating the potential benefits of FD receiver jamming techniques in enhancing information security from a \emph{network} perspective is an interesting, but much more sophisticated issue, since we should take into account numerous interferers and eavesdroppers that are randomly distributed over the network. In addition, using FD receiver jamming in a network is confronted with two fundamental challenges as follows:
\begin{itemize}
	\item Theoretically,
	activating too many FD receivers to send jamming signals brings severe self- and mutual-interference to legitimate receivers, thus impairing the reliability of the ongoing information transmission.
	This will result in few secure links being established and accordingly the poor secrecy throughput.
	\item Practically, employing FD receivers incurs more system cost and overhead. FD transceivers are more expensive than half-duplex (HD) transceivers. From energy efficiency perspective,
	more circuit power is consumed to enable the FD operation or to mitigate the self-interference caused by FD radios, which leads to low energy efficiency.
\end{itemize}

Motivated by these, a proper way to deploy FD receivers is to make a portion of legitimate receivers work in the FD mode simultaneously sending jamming signals and receiving desired signals, and make the rest work in the HD mode just receiving desired signals. This results in a hybrid full-/half-duplex receiver deployment strategy.
Then, a question is naturally raised: \emph{What should be the optimal fraction of FD receivers in order to optimize the network security performance?}
To the best of our knowledge, this question has not been answered by  existing literature. So far, a fundamental analysis on the network security performance, in aspects like \emph{ASLN}, \emph{NST} and \emph{NSEE}, is still lacking for a wireless ad hoc network with hybrid full-/half-duplex receivers.
This motivates our work.

\subsection{Our Work and Contributions}
In this paper, we study physical layer security for a wireless ad hoc network under a stochastic geometry framework \cite{Haenggi2009Stochastic}.
Each transmitter in this network is equipped with a single antenna and sends a secret message to an intended single-antenna receiver, in the presence of randomly located multi-antenna eavesdroppers. A hybrid full-/half-duplex receiver deployment strategy is proposed, where a fraction of legitimate receivers work in the FD mode receiving desired signals and radiating jamming signals simultaneously, and the remaining receivers work in the HD mode just receiving desired signals.
The main contributions of this paper are summarized as follows:
\begin{itemize}
	\item
	We investigate the fundamental tradeoff between secrecy and reliability via jammers.
	We analyze the connection outage probability and the secrecy	outage probability of a typical legitimate link, and provide both accurate expressions and tractable approximations
	for them.
	\item
	We study three important performance metrics on network security, namely, ASLN, NST and NSEE, respectively.
	We prove that these metrics are all \emph{quasi-concave} functions of the fraction of FD receivers, and derive the optimal deployment fractions to maximize them.
	\item
	We further develop insights into the behavior of the optimal fraction of FD receivers with respect to various network parameters.
	We also provide closed-form expressions for this optimal fraction in special cases, e.g., under a perfect SIC assumption or in a dense network.
\end{itemize}

%It is worth mentioning that the fundamental tradeoff between secrecy and reliability has also been discussed in \cite{Zhou2011Throughput}. However, our work  distinguishes from \cite{Zhou2011Throughput}, not only in the analysis but also in some new results related to network design. For example, in \cite{Zhou2011Throughput} the density of legitimate nodes is designed to maximize NST, whereas here we design the allocation between FD and HD receivers to maximize ASLN, NST and NSEE. Generally, to switch a node's work mode between FD and HD modes is more convenient and flexible and perhaps causes a lower network load than to change the network scale. What's more, the authors in \cite{Zhou2011Throughput} only provide the optimal density in the low throughput regime, whereas here we derive the optimal fraction of FD receivers for a more general scenario. Therefore we are able to develop new insights into the behaviors of the optimal fraction with respect to various parameters that will significantly affect network security but are ignored in \cite{Zhou2011Throughput}, e.g., the distance between a legitimate pair and the required connection outage probability.

\subsection{Organization and Notations}
The remainder of this paper is organized as follows.
In Section II, we describe the system model.
In Section III, we analyze the connection outage and secrecy outage probabilities of an arbitrary legitimate link.
In Sections IV, V and VI, we optimize the fraction of FD receivers to maximize ASLN, NST and NSEE, respectively.
In Section VII, we conclude our work.

\emph{Notations}:
bold uppercase (lowercase) letters denote matrices (vectors).
$(\cdot)^H$, $(\cdot)^{-1}$,  $\mathrm{Pr}\{\cdot\}$, and $\mathbb{E}_A(\cdot)$ denote Hermitian transpose, inversion, probability, and the  expectation of $A$, respectively.
$\mathrm{CN}(\mu, \nu)$ denotes the circularly symmetric complex Gaussian distribution with mean $\mu$ and variance $\nu$.
$\ln(\cdot)$ denotes the natural logarithm.
$f^{'}(q)$ and $f^{''}(q)$ denote the first- and second-order derivatives of $f(q)$ on $q$, respectively.
$\mathbf{1}_{\mathrm{FD}}(x)$ is an indicator function with $\mathbf{1}_{\mathrm{FD}}(x)=1$ for $x\in\{\mathrm{FD}\}$ and $\mathbf{1}_{\mathrm{FD}}(x)=0$ for $x\notin\{\mathrm{FD}\}$.
${L}_{I}(s)=\mathbb{
	E}_I\left(e^{-sI}\right)$ is the Laplace transform of $I$.
$[x]^{+}\triangleq \max(x,0)$.

\section{System Model}
We consider a wireless ad hoc network composed of numerous single-antenna legitimate transmitter-receiver pairs, coexisting with randomly located $N_e$-antenna eavesdroppers.
Each legitimate transmitter sends a secret message to its paired receiver located a distance
$r_o$ away\footnote{The assumption of a common legitimate link distance is quite generic in analyzing a wireless ad hoc network  \cite{Zhou2011Throughput,Zhang2013Enhancing}, which eases the mathematical analysis.
	Nevertheless, in principle the obtained results can be  generalized to an arbitrary distribution of $r_o$ \cite{Haenggi2009Stochastic}.}.
We assume that a fraction $q$ of legitimate receivers work in the FD mode such that each of them simultaneously receives the desired signal and radiates a jamming signal to confuse eavesdroppers, and the others work in the HD mode only receiving desired signals.
Legitimate receivers and eavesdroppers are distributed according to
independent homogeneous Poisson point processes (PPPs) \cite{Zheng2014Transmission} $\Phi_l$ with density $\lambda_l$ and $\Phi_e$ with density $\lambda_e$, respectively.
Using the property of thinning for a PPP, the distributions of HD and FD receivers follow independent PPPs $\Phi^{\text{HD}}$ with density $\lambda^{\text{HD}} = (1-q)\lambda_l$ and $\Phi^{\text{FD}}$ with density $\lambda^{\text{FD}} = q\lambda_l$, respectively.
We denote by $\tilde \Phi^\text{HD}$ and $\tilde \Phi^\text{FD}$ the location sets of the transmitters corresponding to HD and FD receivers, respectively.
According to the displacement theorem \cite[page 35]{Haenggi2012Stochastic},
$\tilde \Phi^\text{HD}$ and $\tilde \Phi^\text{FD}$ are also independent PPPs with densities $\lambda^\text{HD}$ and $\lambda^\text{FD}$, respectively.
We use $r_{xy}$ to denote the distance between a node located at $x$ and a node at $y$.
For convenience, we use $\tilde x$ to denote the location of a transmitter whose paired receiver is located at $x$.

Wireless channels, including legitimate channels and wiretap channels, are assumed to suffer a large-scale path loss governed by the exponent $\alpha>2$ together with a quasi-static Rayleigh fading with fading coefficients independent and identically distributed (i.i.d.) obeying $\mathrm{CN}(0,1)$.
%Since each eavesdropper receives signals passively, its channel state information is not known, whereas its channel statistics information is assumed to be available\footnote{This hypothesis is quite generic and has been widely adopted in the literature on physical layer security, e.g.,
%\cite{Wang2015Enhancing}-\cite{Goel2008Guaranteeing}.}.
Due to uncoordinated concurrent transmissions, the aggregate interference at a receiver dominates the thermal noise.
Thereby, we concentrate on an interference-limited scenario by ignoring thermal noise, given that the inclusion of thermal noise results in a more complicated analysis but provides no significant qualitative difference.
% Nevertheless, our results can be easily generalized to the case with thermal noise.

Throughout this paper, we denote $\mathbf{S}\in\{\text{HD},\text{FD}\}$ by default.
For convenience, the legitimate link with an $\mathbf{S}$ receiver is called an $\mathbf{S}$-link.
Considering a typical $\mathbf{S}$ receiver located at the origin $o$, its signal-to-interference ratio (SIR) is given by
\begin{align}\label{sir_S}
\gamma_o^\mathbf{S} = \frac{P_th_{\tilde o o}r_o^{-\alpha}}
{I^\text{HD}+I^\text{FD}+\mathbf{1}_{\mathrm{FD}}(\mathbf{S})\eta P_j},
\end{align}
where $I^\text{HD}\triangleq \sum_{\tilde x\in\tilde \Phi^\text{HD}}P_t h_{\tilde x o} r_{\tilde x o}^{-\alpha}$,
$I^\text{FD}\triangleq\sum_{\tilde x\in\tilde \Phi^\text{FD}\setminus \tilde o}\left(P_t h_{\tilde x o} r_{ \tilde x o}^{-\alpha}+P_j h_{x o}r_{ xo}^{-\alpha}\right)$
and $\eta P_j $ denote the interferences from HD links, from FD links and from the typical FD receiver itself, respectively;
$P_t$ and $P_j$ denote the transmit powers of a legitimate transmitter and of an FD receiver, respectively;
$h_{x y}$ denotes the fading channel gain obeying $\mathrm{Exp(1)}$;
$\eta$ is a parameter that reflects the SIC capability, and $\eta=0$ refers to a perfect SIC while $0<\eta\le 1$ corresponds to different levels of SIC.
Note that
$r_{\tilde x o}$ and $r_{x o}$ in $I^\text{HD}$ and $I^\text{FD}$ are correlated, and they satisfy  $r_{\tilde x o}=\sqrt{r_{x o}^2+r_o^2-2r_{x o}r_o\cos\theta_x}$, where the angle $\theta_x$ is uniformly distributed in the range $[0, 2\pi]$.

For eavesdroppers, we consider a worst-case wiretap scenario where each eavesdropper has multiuser decoding ability and adopts a successive interference cancellation minimum mean square
error (MMSE) receiver.
The eavesdropper located at $e$ is able to decode and cancel undesired information signals and uses the MMSE detector
%Consider the eavesdropper located at $e$, and the weight vector of it has the form of
\begin{equation}\label{weight_mmse}
\bm w^{\mathbf{S}}_{e} =  \left(\bm R^{\mathbf{S}}_{e}\right)^{-1}\bm g_{\tilde oe}
\end{equation}
to aggregate the desired signal, where $\bm R^{\mathbf{S}}_{e} \triangleq \sum_{x\in\Phi^{\text{FD}}\setminus o}{P_j\bm g_{xe}\bm g_{xe}^H}{ r_{x e}^{-\alpha}}
+\mathbf{1}_{\mathrm{FD}}(\mathbf{S})\bm R_{oe}$ with $\bm R_{oe}\triangleq P_j\bm g_{oe}\bm g_{oe}^H r_{o e}^{-\alpha}$, and
$\bm g_{x e}$ denotes the $N_e\times 1$ complex fading coefficient vector related to the link from a node at $x$ to an eavesdropper at $e$.
The corresponding SIR of the eavesdropper is
\begin{equation}\label{sir_e}
\gamma^{\mathbf{S}}_{e} = P_t\bm g_{\tilde oe}^H(\bm R^{\mathbf{S}}_{e})^{-1}\bm g_{\tilde oe}r_{\tilde o e}^{-\alpha}.
\end{equation}

\subsection{Secrecy Performance Metrics}
We assume eavesdroppers do not collude with each other such that each of them individually decodes a secret message.
To guarantee secrecy, each legitimate transmitter adopts the Wyner's wiretap encoding scheme \cite{Wyner1975Wire-tap} to encode secret information.
Thereby, two types of rates, namely, the rate of transmitted codewords $R_t$ and the rate of embedded information bits ${R}_{s}$, need to be designed to meet requirements in terms of the connection outage and secrecy outage probabilities.

\begin{itemize}
	
	\item
	\emph{Connection outage probability}.
	If a legitimate $\mathbf{S}$-link can support rate ${R}_t$, the legitimate receiver is able to decode a secret message and perfect connection is assured in this link; otherwise a connection outage occurs.
	The probability that such a connection outage event takes place is referred to as the connection outage probability, denoted as $p_{co}^{\mathbf{S}}$.
	
	\item
	\emph{Secrecy outage probability}.
	In the Wyner's wiretap encoding scheme, the rate redundancy
	${R}_e\triangleq {R}_t -{R}_s$ is exploited to provide secrecy against eavesdropping.
	If the value of ${R}_e$ lies above the capacity of the most detrimental eavesdropping link,	
	no information is leaked to eavesdroppers and perfect secrecy is promised in the legitimate link \cite{Wyner1975Wire-tap}; otherwise a secrecy outage occurs.
	The probability that such a secrecy outage event takes place in an $\mathbf{S}$-link is referred to as the secrecy outage probability, denoted as $p_{so}^{\mathbf{S}}$.
	
\end{itemize}

In this paper, we concern ourselves with the following three important metrics that measure the network-wide security performance from an outage perspective.

1) \emph{ASLN}.
A link in which neither connection outage nor secrecy outage occurs is called a secure link \cite{Ma2015Interference}.
To measure how many secure links can be guaranteed under rates $R_t$ and $R_s$, we use the metric named ASLN, which is defined as the average number of secure links per unit area.
Due to the independence of $p_{co}^{\text{S}}$ and $p_{so}^{\text{S}}$, ASLN, denoted as $\bm{N}$, is mathematically given by
\begin{equation}\label{sl_def}
\bm{N}\triangleq q\lambda_l(1-p_{co}^{\text{FD}})(1-p_{so}^{\text{FD}})
+
(1-q)\lambda_l(1-p_{co}^{\text{HD}})(1-p_{so}^{\text{HD}}).
\end{equation}

2) \emph{NST}.
To assess the efficiency of secure transmissions, we use the metric named NST \cite{Zhou2011Throughput}, which is defined as the achievable rate of successful information transmission per unit area under the required  connection outage and secrecy outage  probabilities.
The NST, denoted as $\mathbf{\Omega}$, under a connection outage probability   $p_{co}^{\mathbf{S}}=\sigma$ and a secrecy outage probability $p_{so}^{\mathbf{S}}=\epsilon$ is given by
\begin{equation}\label{st_def}
\mathbf{\Omega}\triangleq q\lambda_l(1-\sigma){R}^{\text{FD}}_s
+
(1-q)\lambda_l(1-\sigma){R}^{\text{HD}}_s,
\end{equation}
where ${R}^{\mathbf{S}}_s\triangleq [{R}^{\mathbf{S}}_t-{R}^{\mathbf{S}}_e]^+$, with ${R}^{\mathbf{S}}_t$ and ${R}^{\mathbf{S}}_e$ the codeword rate and redundant rate that satisfy $p_{co}^{\mathbf{S}}(R^{\mathbf{S}}_t)=\sigma$
and $p_{so}^{\mathbf{S}}(R^{\mathbf{S}}_e)=\epsilon$, respectively.
The unit of $\mathbf{\Omega}$ is  $\mathrm{nats/s/Hz/m^2}$.

3) \emph{NSEE}.
To evaluate the energy efficiency of secure transmissions, we use the metric named NSEE, denoted as $\mathbf{\Psi}$, which is defined as the ratio of NST to the power consumed per unit area,
\begin{equation}\label{ee_def}
\mathbf{\Psi} \triangleq \frac{\mathbf{\Omega}}{\lambda_l(P_t+P_c)+q\lambda_lP_j},
\end{equation}
where $P_c$ combines the dynamic circuit power consumption of transmit chains and the static power consumption in transmit modes \cite{Ha2013Energy}.
The unit of $\mathbf{\Psi}$ is $\mathrm{nats/Joule/Hz}$.

We emphasize that, the fraction $q$ of FD receivers triggers a non-trivial trade-off between reliability and secrecy, and plays a key role in improving the metrics given above.
Intuitively, under a larger $q$, more FD jammers are activated against eavesdroppers which benefits the secrecy; whereas the increased jamming signals also interfere with legitimate receivers and thus harm the reliability.
The overall balance of such conflicting effects needs to be carefully addressed.
In Sections IV, V and VI, we are going to respectively determine the optimal fraction $q$ that

\begin{itemize}
	
	\item
	maximizes ASLN $\bm{N}$ given a pair of wiretap code rates $R_t$ and $R_s$;
	
	\item
	maximizes NST $\mathbf{\Omega}$ given a pair of outage probabilities $\sigma$ and $\epsilon$;
	
	\item
	maximizes NSEE $\mathbf{\Psi}$ with and without considering a minimum required NST.
	
\end{itemize}
Before proceeding, in the following section we first provide some insights into the behavior of the connection outage probability $p_{co}^{\mathbf{S}}$ and the secrecy outage probability $p_{so}^{\mathbf{S}}$ with respect to network parameters like $q$, $\eta$, etc., which is very important to subsequent network design.

\section{Outage Probability Analysis}
In this section, we derive the connection outage probability and the secrecy outage probability for an arbitrary legitimate link.
For ease of notation, we define $\delta\triangleq {2}/{\alpha}$,
$\kappa\triangleq {\pi\Gamma\left(1+\delta\right)\Gamma\left(1-\delta\right)}$ and $\rho\triangleq {P_j}/{P_t}$, which will be used throughout the paper.

\subsection{Connection Outage Probability}
The connection outage probability of a typical $\mathbf{S}$-link is defined as the probability that the SIR  $\gamma_o^\mathbf{S}$ given in \eqref{sir_S} falls below an SIR threshold $\tau_t\triangleq2^{R_t}-1$, i.e.,
\begin{equation}\label{pco_def}
p_{co}^\mathbf{S}\triangleq \mathrm{Pr}\{\gamma_o^\mathbf{S}< \tau_t\}.
\end{equation}
The general expression of $p^\mathbf{S}_{co}$ is provided by the following theorem.
The interested readers are referred to \cite[Th. 1]{Zheng2017Safeguarding} for a detailed proof.

\begin{theorem}\label{pco_exact_theorem}
	The connection outage probability of a typical $\mathbf{S}$-link is given by
	\begin{align}\label{pco_exact}
	p^\mathbf{S}_{co} = 1-e^{-\mathbf{1}_{\mathrm{FD}}(\mathbf{S})\rho\eta r_o^{\alpha}\tau_t}
	e^{-\kappa(1-q)\lambda_l
		r_o^2\tau_t^{\delta}}
	{L}_{I^\text{FD}}\left(
	{r_o^\alpha\tau_t}/{P_t}\right),
	\end{align}
	where
	${L}_{I^\text{FD}}({r_o^\alpha\tau_t}/{P_t})
	=\exp\Big(-q\lambda_l\int_0^{\infty}
	\int_0^{2\pi}	\Big(1-
\frac{1}{1+r_o^\alpha\tau_tv^{-\alpha}}
	\frac{1}
	{1+r_o^\alpha\tau_t\left(v^2+r_o^2-2vr_o \cos\theta\right)^{-\alpha/2}}
	\Big)vd\theta dv\Big)$.
\end{theorem}
%\begin{proof}
%	Please refer to Appendix \ref{appendix_pco_exact_theorem}.
%\end{proof}

Theorem \ref{pco_exact_theorem} provides an exact connection outage probability with  three parts $e^{-\mathbf{1}_{\mathrm{FD}}(\mathbf{S})\rho\eta r_o^{\alpha}\tau_t}$, $e^{-\kappa(1-q)\lambda_lr_o^2\tau_t^{\delta}}$ and ${L}_{I^\text{FD}}\left({r_o^\alpha\tau_t}/{P_t}\right)$, reflecting the impacts of the interferences from the typical receiver itself, from HD links and from FD links, respectively.
Although with $p^\mathbf{S}_{co}$ given in \eqref{pco_exact} we no longer need to execute time-consuming Monte Carlo simulations, the double integral in $ {L}_{I^\text{FD}}({r_o^\alpha\tau_t}/{P_t})$ greatly complicates the further analysis, which motivates a more compact form.
In the following theorem, we provide the closed-form upper and lower bounds for $p^\mathbf{S}_{co}$, and refer the interested readers to \cite[Th. 2]{Zheng2017Safeguarding} for a detailed proof.
\begin{theorem}\label{pco_bound_theorem}
	Connection outage probability	 $p^\mathbf{S}_{co}$ is upper and lower bounded respectively by %$p^{\mathbf{S},\mathrm{U}}_{co}$ and  $p^{\mathbf{S},\mathrm{L}}_{co}$ given below,
	\begin{align}
	\label{pco_upper}
	p^{\mathbf{S},\mathrm{U}}_{co}&= 1-e^{-\mathbf{1}_{\mathrm{FD}}(\mathbf{S})\rho\eta r_o^{\alpha}\tau_t}
	e^{-\kappa
		r_o^2\tau_t^{\delta}\lambda_l\left(1+\rho^\delta q\right)},
	\\
	\label{pco_lower}
\!\!\!	p^{\mathbf{S},\mathrm{L}}_{co}&= 1-e^{-\mathbf{1}_{\mathrm{FD}}(\mathbf{S})\rho\eta r_o^{\alpha}\tau_t}
	e^{-\kappa
		r_o^2\tau_t^{\delta}\lambda_l\left(1+\frac{(1+\delta)\rho^{\delta}-(1-\delta)}{2}q\right)}.\!\!\!
	\end{align}
\end{theorem}
%\begin{proof}
%	Please refer to Appendix \ref{appendix_pco_bound_theorem}.
%\end{proof}

Theorem \ref{pco_bound_theorem} shows that both bounds for the connection outage probability increase exponentially in $\eta$, $q$ and $\lambda_l$, because of the increase of self- and mutual-interference.
The relationships between the connection outage probability $p^{\text{FD}}_{co}$ and parameters $q$ and $\eta$ are validated in Fig. \ref{PCO}, where  the results labeled by $\eta=0$ also refer to an HD counterpart.
We observe that, although $p^{\text{FD}}_{co}$ increases as $q$ increases, the effect is not very remarkable when $\eta$ is large.
This is because, in
the large $\eta$ region the self-interference perceived at an FD receiver dominates the interference (including both undesired and jamming signals) from the other network nodes.

\begin{figure}[!t]
	\centering
	\includegraphics[width=3.0in]{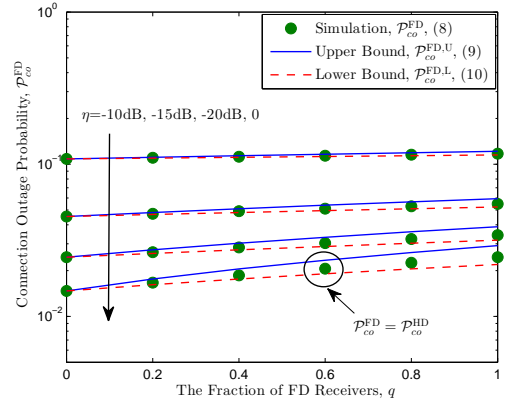}
	\caption{Connection outage probability vs. $q$ for different values of $\eta$, with $\alpha=4$, $\rho=1$, $r_o=1$, $\lambda_l = 3\times 10^{-3}$ and $\tau_t = 1$.}
	\label{PCO}
\end{figure}

\subsection{Secrecy Outage Probability}
The secrecy outage probability of a typical $\mathbf{S}$-link is defined as the complement of the probability that any eavesdropper's SIR $\gamma^{\mathbf{S}}_{e}$ falls below an SIR threshold $\tau_e\triangleq2^{R_e}-1$, i.e.,
	\begin{equation}\label{pso_fd_def}
	p^{\mathbf{S}}_{so}\triangleq 1-\mathbb{E}_{\Phi^{\text{FD}}}\mathbb{E}_{\Phi_e}
	\left[\prod_{e\in\Phi_e}\mathrm{Pr}\left\{
	\gamma^{\mathbf{S}}_{e}
	<\tau_{e}|\Phi_e,\Phi^{\text{FD}}\right\}\right].
	\end{equation}
	To calculate exact $p^{\mathbf{S}}_{so}$ is very difficult. Instead, we give an upper bound for $p^{\mathbf{S}}_{so}$ in the following theorem.
	Please refer to \cite[Th. 3]{Zheng2017Safeguarding} for a detailed proof.
	
	\begin{theorem}\label{pso_exact_theorem}
		Secrecy outage probability $p^{\mathbf{S}}_{so}$ of a typical $\mathbf{S}$-link is upper bounded by
		\begin{equation}\label{pso_exact}
		p^{\mathbf{S},\mathrm{U}}_{so} = 1 - \exp\left(
		-\lambda_e\sum_{n=0}^{N_e-1}
		\sum_{i=0}^{\min(n,1)}
		\frac{\left(\kappa\lambda^{\text{FD}}
			\rho^{\delta}\tau_e^{\delta}\right)
			^{n-i}}
		{(n-i)!}\Xi_{n,i}^{\mathbf{S}}\right),
		\end{equation}
		where $\Xi_{n,i}^{\mathbf{S}}=\int_0^{\infty}\int_0^{2\pi}
		\Lambda^{\mathbf{S}}_{i,\theta,v}v^{2(n-i)}
		e^{-\kappa\lambda^{\text{FD}}
			\rho^{\delta}\tau_e^{\delta}v^2}d\theta vdv$ with $\Lambda^{\mathbf{S}}_{i,\theta,v}=\frac{\left(\mathbf{1}_{\mathrm{FD}}(\mathbf{S})\rho\tau_e\left({v}/
			{\sqrt{v^2+r_o^2-2vr_o\cos\theta}}\right)^{\alpha}\right)^i}
		{1+\mathbf{1}_{\mathrm{FD}}(\mathbf{S})\rho\tau_e\left({v}/
			{\sqrt{v^2+r_o^2-2vr_o\cos\theta}}\right)^{\alpha}}$.
	\end{theorem}
%	\begin{proof}
%		Please refer to Appendix \ref{appendix_pso_exact_theorem}.
%	\end{proof}
	
	In the following sections, we use upper bound $p^{\mathbf{S},\mathrm{U}}_{so}$ to replace exact $p^{\mathbf{S}}_{so}$, not simply for a tractable analysis but also for the following two reasons: on one hand, $p^{\mathbf{S},\mathrm{U}}_{so}$ provides a pessimistic evaluation of secrecy performance, which actually benefits a robust design; on the other hand, as \cite{Zhou2011Throughput} shows, $p^{\mathbf{S},\mathrm{U}}_{so}$ will converge to $p^{\mathbf{S}}_{so}$ at the low secrecy outage probability regime, and a low secrecy outage probability is expected in order to guarantee a high level of secrecy.

Clearly, since $\mathbf{1}_{\mathrm{FD}}(\text{HD})=0$, we have $\Lambda^{\text{HD}}_{0,\theta,v}=1$ and $\Lambda^{\text{HD}}_{i,\theta,v}=0$ for $i>0$.
Substituting these results into \eqref{pso_exact} yields a closed-form expression for $p^{\text{HD}}_{so} $ given below,
\begin{align}\label{pso_hd_exact}
p^{\text{HD}}_{so} &= 1 - \exp\Bigg(
-\pi\lambda_e\sum_{n=0}^{N_e-1}
\frac{\left(\kappa\lambda^{\text{FD}}
	\rho^{\delta}\tau_e^{\delta}\right)
	^{n}}
{n!}\int_0^{\infty}
 v^{2n}\times\nonumber\\
&e^{-\kappa\lambda^{\text{FD}}\rho^{\delta}\tau_e^{\delta}v^2}dv^2\Bigg)=1-e^{-\frac{\pi\lambda_eN_e}{\kappa q\lambda_l\rho^{\delta}\tau_e^{\delta}}},
\end{align}
where the last equality follows from formula
%the formula $\int_0^{\infty}v^{n}e^{-\mu v}dv=({n!})/
%({\mu ^{n+1}})$
\cite[(3.381.4)]{Gradshteyn2007Table}.
As to $p^{\text{FD}}_{so} $, the double integral in \eqref{pso_exact} makes it difficult to analyze.
Given that a single-antenna transmitter in a large-scale ad hoc network usually has low transmit power and very limited coverage, we should set the legitimate link distance $r_o$ sufficiently small (compared with the distance between two nodes that are not in pair) to guarantee both reliability and secrecy.
		In the following, we resort to an asymptotic analysis by letting $r_o\rightarrow 0$ in \eqref{pso_exact} in order to develop useful and tractable insights into the behavior of $p^{\text{FD}}_{so} $.
	The following corollary gives a quite simple approximation for $p^{\text{FD}}_{so} $.
	\begin{corollary}\label{pso_approx_corollary}
		In the small $r_o$ regime, i.e., $r_o\rightarrow 0$,  $p^{\text{FD}}_{so} $ in \eqref{pso_exact} is approximated by
		\begin{equation}\label{pso_approx}
		\tilde p^{\text{FD}}_{so}
		= 1-\exp\left(-\frac{\pi\lambda_eN_e}
		{\kappa q\lambda_l\rho^{\delta}
			\tau_e^{\delta}}\left(
		1-\frac{\rho\tau_e/N_e}{1+\rho\tau_e}\right)\right).
		\end{equation}
	\end{corollary}
	\begin{proof}
		Recalling Theorem \ref{pso_exact_theorem}, plugging $r_o\rightarrow 0$ into $\Lambda^{\text{FD}}_{i,\theta,v}$ yields	$\Lambda^{\text{FD}}_{i,\theta,v}=
		\frac{\left(\rho\tau_e\right)^i}
		{1+\rho\tau_e}$, and thus $\Xi_{n,i}^{\text{FD}}=\frac{2\pi(\rho\tau_e)^i(n-i)!}{(1+\rho\tau_e)}(\kappa\lambda^{\text{FD}}
		\rho^{\delta}\tau_e^{\delta})^{i-n-1}$. Substituting $\Xi_{n,i}^{\text{FD}}$ into \eqref{pso_exact} completes the proof.
	\end{proof}
	
%	\begin{table}
%		\caption{$\Delta_{p_{so}}\triangleq ({p^{\text{FD}}_{so}-\tilde p^{\text{FD}}_{so}})/{p^{\text{FD}}_{so}} $~ vs.~ $r_o~\&~ \Delta\lambda\triangleq \lambda_e/ \lambda_l$}
%		\begin{center}
%			\begin{tabular}{|c|c|c|c|c|c|c|c|}
%				\hline
%				\diagbox {$r_o$}{$\Delta_{p_{so}}( \%)$}{$\Delta\lambda$}
%				&  0.01 & 0.02 & 0.05 & 0.1 & 0.2 &  0.5 &1  \\\hline
%				1 &  0.08 & 0.08 &0.07 &  0.05 & 0.03 &0 & 0   \\\hline
%				5 &  0.12 & 0.12 &0.10 &  0.08 & 0.04 &0.01 & 0   \\\hline
%				10 &0.49 & 0.47 &0.40 &0.31 & 0.17 &0.02 & 0  \\\hline
%				%			20 &  3.52 & 3.40 &3.04 &  2.51 & 1.66 &0.38 & 0.02 & 0  \\\hline
%				%			30 &  6.08 & 5.87 &5.26 &  4.34 & 2.86 &0.65 & 0.03 & 0  \\\hline
%				30 &  3.76 & 3.58 &3.08 &  2.35 & 1.30 &0.15 & 0  \\\hline
%				%			100 &  6.02 & 5.73 &4.92 &  3.75 & 2.06 &0.23 & 0  \\\hline
%			\end{tabular}
%		\end{center}
%		\label{Ps_r}
%	\end{table}
	
We stress that, although Corollary \ref{pso_approx_corollary} is established under the assumption $r_o\rightarrow 0$, it actually applies to more general scenarios.
Fig. \ref{PSO} shows that $\tilde p^{\text{FD}}_{so}$ in \eqref{pso_approx} approximates to $p_{so}^{\mathrm{FD}}$ in \eqref{pso_exact} in quite a wide range of $r_o$ and $\lambda_e$ particularly when $\lambda_f$ is small, which demonstrates high accuracy for the approximation.
Hereafter, unless specified otherwise, we often use this approximation to deal with the secrecy outage probability.

\begin{figure}[!t]
	\centering
	\includegraphics[width=3.0in]{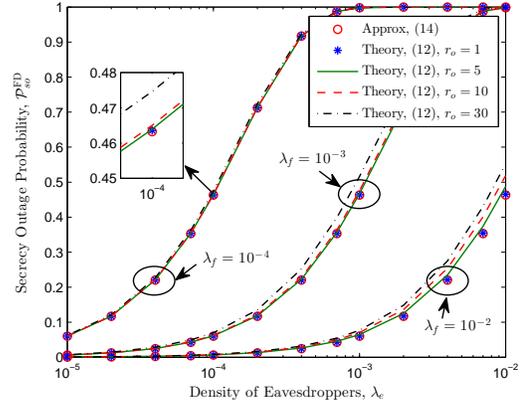}
	\caption{Secrecy outage probability vs. $\lambda_e$ for different values of $r_o$ and $\lambda_f$, with $\alpha=4$, $\rho=10$ and $\tau_e = 1$.}
	\label{PSO}
\end{figure}
	
	Eqn. \eqref{pso_hd_exact} and \eqref{pso_approx} clearly show that secrecy outage probabilities increase exponentially with $\lambda_e$ and $N_e$.
	This is ameliorated by increasing $q$ or $\rho$.
	In addition, secrecy outage probabilities increase as $\alpha$ increases.
	This is because, in a large path-loss exponent environment, jamming signals have undergone a strong attenuation before they arrive at  eavesdroppers.

\section{Area secure link number}

In this section, we maximize ASLN $\bm{N}$ under a given pair of wiretap code rates $R_t$ and $R_s$ by determining the optimal fraction $q$ of FD receivers.

To facilitate a robust  design, we use the upper bounded connection outage probability  $p_{co}^{\mathbf{S},\mathrm{U}}$ given in \eqref{pco_upper}, which actually pessimistically assess the connection performance.
We also suppose eavesdroppers use a large number of antennas in order to do better wiretapping, which also gives a pessimistic evaluation of the secrecy performance.
Resorting to an asymptotic analysis of $p^{\text{FD}}_{so}$ by letting $N_e\gg 1$ in \eqref{pso_approx}, $p^{\text{FD}}_{so}$ shares the same expression as
$p^{\text{HD}}_{so}$ in \eqref{pso_hd_exact}, i.e.,
\begin{equation}\label{pso_large}
p_{so}^{\mathbf{S}}=1-e^{-\frac{\pi\lambda_eN_e}{\kappa q\lambda_l\rho^{\delta}\tau_e^{\delta}}}.
\end{equation}
Substituting \eqref{pco_upper} and \eqref{pso_large} into \eqref{sl_def} yields
\begin{align}\label{sl_exp}
\bm{N}
&=\lambda_l
\left(q e^{-\rho\eta r_o^{\alpha}\tau_t}+1-q\right)
e^{-\kappa r_o^2\tau_t^{\delta}\lambda_l(1+\rho^{\delta}q)-\frac{\pi\lambda_eN_e}{\kappa q\lambda_l\rho^{\delta}\tau_e^{\delta}}}.
\end{align}
Introducing an auxiliary function $ {F}(q)=(qA+(1-q))e^{-Bq-C/q}$ with ${A}\triangleq e^{-\rho\eta r_o^{\alpha}\tau_t}<1$, ${B}\triangleq\kappa r_o^2\tau_t^{\delta}
\rho^{\delta}\lambda_l$ and ${C}\triangleq \frac{\pi\lambda_eN_e}{\kappa \lambda_l\rho^{\delta}\tau_e^{\delta}}$, we have $\bm{N}=\lambda_l e^{-\kappa r_o^2\tau_t^{\delta}\lambda_l} {F}(q)$ such that parameter $q$ only exists in $ {F}(q)$.
Hence, maximizing $ \bm{N}$ is equivalent to maximizing $ {F}(q)$, which can be formulated as
%\begin{subequations}
\begin{align}\label{sl_max_def}
&\max_{q}~  {F}(q)=(q{A}+(1-q))e^{-{B}q-{C}/q},~~\mathrm{s.t.}~ 0<q\le 1.
\end{align}
%\end{subequations}
%It seems that solving problem \eqref{sl_max_def} requires a one dimensional exhaustive searching.
In the following theorem,
we prove the quasi-concavity \cite[Sec. 3.4.2]{Boyd2004Convex} of $ {F}(q)$ in $q$, and give the optimal solution of problem \eqref{sl_max_def}.
\begin{theorem}\label{opt_q_sl_theorem}
	The optimal fraction of FD receivers that maximizes ASLN $\bm{N}$ is given by
	\begin{equation}\label{opt_q_sl}
	q_{sl}^* =
	\begin{cases}
	1,&  {\pi\lambda_eN_e}
	> {\kappa\lambda_l\rho^{\delta}\tau_e^{\delta}}(1/A+B-1),\\
	q_{sl}^{\circ},& \text{otherwise},
	\end{cases}
	\end{equation}
	where $q_{sl}^{\circ}$ is the unique root $q$ of the following equation
	\begin{equation}\label{opt_sl_q_exp}
	\left(A+q^{-1}-1\right)\left(1+{C}{q^{-1}}-Bq\right)-q^{-1}=0.
	\end{equation}
	The left-hand side (LHS) of \eqref{opt_sl_q_exp} is initially positive and then negative when $C\le 1/A+B-1$; and thus, $q_{sl}^{\circ}$ can be efficiently calculated using the bisection method with \eqref{opt_sl_q_exp}.	
\end{theorem}
\begin{proof}
	Please refer to Appendix \ref{appendix_opt_q_sl_theorem}.
\end{proof}

Theorem \ref{opt_q_sl_theorem} indicates that as eavesdropper density $\lambda_e$ or eavesdropper antenna number $N_e$ is sufficiently large such that ${\pi\lambda_eN_e}
> {\kappa\lambda_l\rho^{\delta}\tau_e^{\delta}}(1/A+B-1)$, all legitimate receivers should work in the FD mode; otherwise a portion of HD receivers are permitted, just as depicted in Fig. \ref{OPT_Q_SL}.

\begin{figure}[!t]
	\centering
	\includegraphics[width=3.0in]{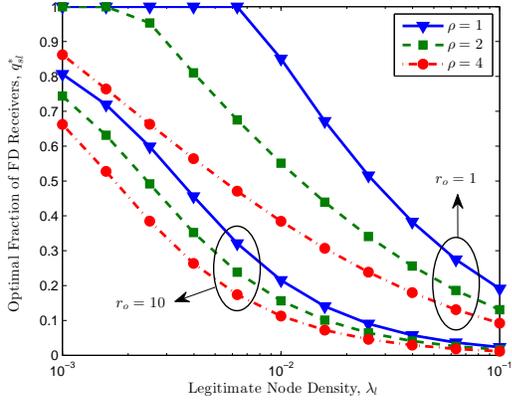}
	\caption{The optimal fraction of FD receivers that maximizes ASLN $\bm{N}$ vs. $\lambda_l$ for different values of $r_o$ and $\rho$, with $\alpha=4$,  $\lambda_e = 10^{-3}$, $N_e=6$, $\eta=-10$dB, $\tau_t=2$ and $\tau_e = 1$.}
	\label{OPT_Q_SL}
\end{figure}

Although it is difficult to provide an explicit expression for the optimal  $q_{sl}^{\circ}$ given in \eqref{opt_q_sl}, we are still able to develop some insights into the behavior of $q_{sl}^{\circ}$ in the following corollary.
\begin{corollary}\label{corollary_q_sl}
	The optimal $q_{sl}^{\circ}$ given in \eqref{opt_q_sl} monotonically increases with $\lambda_e$ and $N_e$, and monotonically decreases with $\lambda_l$, $r_o$, $\eta$, $\rho$, $\tau_t$ and $\tau_e$.
	In the perfect SIC case, i.e., $\eta=0$, a closed-form expression on $q_{sl}^{\circ}$ can be further given by
	\begin{equation}\label{q_sl_sic}
	q^{\circ,\eta=0}_{sl}
	=\sqrt{\frac{C}{B}}
	= \frac{1}{\kappa\lambda_l\rho^\delta r_o}\sqrt{\frac{\pi\lambda_eN_e}{\tau_t^\delta\tau_e^\delta}},
	\end{equation}
	where $q^{\circ,\eta=0}_{sl}$ decreases linearly in $\lambda_l$ and $r_o$, and increases linearly in $\sqrt{\lambda_e}$ and $\sqrt{N_e}$.
\end{corollary}
\begin{proof}
	Please refer to Appendix \ref{appendix_corollary_q_sl}.
\end{proof}

Corollary \ref{corollary_q_sl} provides some useful insights into the optimal fraction of FD receivers, which will benefit network design.
For example, more FD receivers are needed to cope with more eavesdroppers or more eavesdropping antennas; whereas adding legitimate nodes or increasing jamming power allows a smaller fraction of FD receivers.
In addition, we should better activate fewer FD receivers when legitimate link distance $r_o$ increases, since the desired signal suffers a greater attenuation and the negative effect of self-interference increases more significantly.
Some of the properties in Corollary \ref{corollary_q_sl} are verified in Fig. \ref{OPT_Q_SL}, and the others are relatively intuitive.

%Considering a perfect SIC case, i.e., $\eta=0$, a closed-form expression of $q_{sl}^{\circ}$ can be given by
%\begin{equation}\label{q_sl_sic}
%q^{\circ,\eta=0}_{sl}
%=\sqrt{\frac{C}{B}}
%= \frac{1}{\kappa\lambda_l\rho^\delta r_o}\sqrt{\frac{\pi\lambda_eN_e}{\tau_t^\delta\tau_e^\delta}}.
%\end{equation}
%Clearly, the properties provided by Corollary \ref{corollary_q_sl} can be easily extracted
%from \eqref{q_sl_sic}.

Having obtained the optimal fraction $q^*_{sl}$ given in \eqref{opt_q_sl}, the maximum ASLN $\bm{N}^*$ can be calculated by plugging $q^*_{sl}$ into \eqref{sl_exp}.
Fig. \ref{NWSL_COMPARE}
%compares the maximum ASLN obtained at $q=q^*_{sl}$ and the ASLNs obtained at $q=0.1$ and $q=0.5$, respectively.
depicts ASLN as a function of $N_e$ and clearly demonstrates the superiority of our optimization scheme over those fixed-$q$ schemes.
For example, the maximum ASLN obtained at $q=q_{sl}^*$ is nearly twice as large as that obtained at $q=0.5$ for a small $N_e $, and is more than twice as large as that obtained at $q=0.1$ for a large $N_e$.
We observe that, as $\rho$ increases, the ASLN obtained at $q=0.5$ becomes smaller in the small $N_e$ region whereas becomes larger in the large $N_e$ region.
The underlying reason is, when $N_e$ is small, the negative impact of jamming signals on legitimate links is larger than that on wiretap links such that fewer FD jammers should be activated; conversely, as $N_e$ increases, the negative effect of jamming signals on wiretap links increases obviously.
In sharp contrast to this, the maximum $\bm{N}^*$ always increases in $\rho$ regardless of $N_e$.
This is because the optimal fraction $q_{sl}^*$ adaptively decreases as $\rho$ increases so as to mitigate the negative effect of jamming signals.

\begin{figure}[!t]
	\centering
	\includegraphics[width=3.0in]{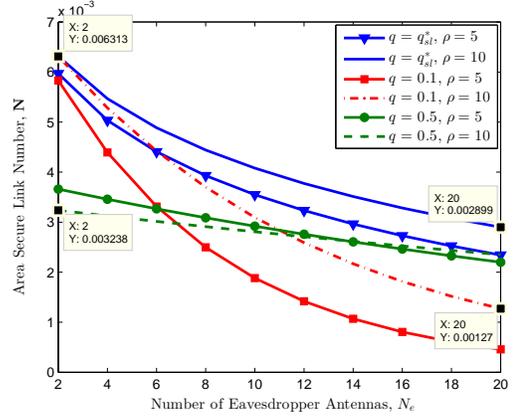}
	\caption{ASLN vs. $N_e$ for different values of $q$ and $\rho$, with $\alpha=3$, $\lambda_l= 10^{-2}$, $\lambda_e = 10^{-3}$, $\eta = -7$dB, $\tau_t=2$ and $\tau_e = 1$.}
	\label{NWSL_COMPARE}
\end{figure}

\section{Network-wide Secrecy Throughput}
In this section, we maximize NST $\mathbf{\Omega}$ under a pair of outage probabilities $\sigma$ and $\epsilon$ by determining the optimal fraction $q$ of FD receivers, which can be formulated as
\begin{align}\label{st_max_def}
\max_{q}~ \mathbf{\Omega},\quad
\mathrm{s.t.}~ 0<q\le 1.
\end{align}
%Note that, we should not set $q$ to zero.
%Otherwise, there is no interference at an eavesdropper and its SIR goes to infinity, since we assume it has the capability of multiuser decoding.
To proceed, we first derive the SIR thresholds $\tau^{\mathbf{S}}_{t}$ and $\tau^{\mathbf{S}}_{e}$ that satisfy $p^{\mathbf{S}}_{co}(\tau^{\mathbf{S}}_{t})=\sigma$
and $p^{\mathbf{S}}_{so}(\tau^{\mathbf{S}}_{e})=\epsilon$, respectively.
We can easily calculate
$\tau^{\text{HD}}_{t}$ from \eqref{pco_upper}; whereas it is in general difficult to derive analytical expressions for $\tau_t^{\text{FD}}$.
However, as reported in \cite{Lee2015Hybrid}, self-interference can be efficiently mitigated by exploiting the propagation domain, analog circuit domain and digital circuit domain; particularly in analog and digital signal processing it is now feasible to have up to 110 dB SIC capability \cite{Bharadia2013Full}.
Such positive news motivates us to consider a perfect SIC case by ignoring self-interference in order to facilitate the design.
Letting $\eta=0$ in \eqref{pco_upper}
and \eqref{pso_large}
yields uniform expressions for $\tau^{\mathbf{S}}_{t}$ and $\tau^{\mathbf{S}}_{e}$, respectively, given by
\begin{equation}\label{beta_t}
\tau^{o}_{t}=\left(\frac{\sigma_o}{\kappa\left(\lambda^{\text{HD}}+[1+\rho^{\delta}]\lambda^{\text{FD}}\right) r_o^2}\right)^{{\alpha}/{2}},
\end{equation}
%We also obtain a uniform expression for $\tau^{\mathbf{S}}_{e}$ from \eqref{pso_large} given by
\begin{equation}\label{beta_e}
\tau^{o}_{e}=\left(\frac{\pi\lambda_eN_e}{\kappa\rho^{\delta}\lambda^{\text{FD}}\epsilon_o}\right)^{{\alpha}/{2}},
\end{equation}
where $\sigma_o\triangleq \ln\frac{1}{1-\sigma}$ and $\epsilon_o\triangleq \ln\frac{1}{1-\epsilon}$.
%However, for the special case $\alpha=4$, the values of $\tau_t^{\text{FD}}$ and $\tau_e^{\text{FD}}$ can be, respectively, given by
%$\tau^{\text{FD},\alpha=4}_{t}=
%\frac{\sqrt{\xi^2+4\rho\eta\ln\frac{1}{1-\sigma}}-\xi }{2\rho\eta r_o^2},$ with
%$\xi\triangleq \kappa\left(\lambda^{\text{HD}}+(1+\rho^{\delta})\lambda^{\text{FD}}\right)$ and
%$\tau^{\text{FD},\alpha=4}_{e}= \frac
% {\left[\left(
% 	\sqrt[3]{12\sqrt{a}-108b}+{b}\right)^2
% 	-{b^2}-12\left(1-\frac{b^2}{3} \right)\right]^2 }
% {{ \rho\left(1-\frac{b^2}{3} \right)}
% 	\left({12\sqrt{a}-108b}\right)^{2/3}
% },$ with $a\triangleq 81\left[\left(\frac{1}{3}-\frac{N_e}{N_e-1} \right )b
% -\frac{2b^3}{27}\right]^2+12\left( 1-\frac{b^2}{3}\right) ^3$ and $b\triangleq\left( {\pi\lambda_e(N_e-1)}\right) /\left({\kappa\lambda^{\text{FD}}	\ln\frac{1}{1-\epsilon}} \right)$.
Substituting $\tau^{\mathbf{S}}_{t}=\tau^{o}_{t}$ and $\tau^{\mathbf{S}}_{e}=\tau^{o}_{e}$ into \eqref{st_def} yields
%\begin{equation}\label{rs_hf}
%	R_s^o=\left[\ln \frac{1+\tau^o_t}{1+\tau^o_e} \right]^+.
%\end{equation}
%substituting which into \eqref{st_def} yields
\begin{align}\label{st_hf}
\mathbf{\Omega}=\lambda_l(1-\sigma) \left[\ln \frac{1+\tau^o_t}{1+\tau^o_e} \right]^+.
\end{align}
%Define $\lambda_a\triangleq \kappa\lambda_lr^2$, $\lambda_b\triangleq \kappa \lambda_lr^2 \left(1+\rho^{\delta}\right)\ge \lambda_a$, and $\lambda_c\triangleq (\kappa \lambda_l\rho^{\delta})/\left( \pi\lambda_e N_e\right)$.
%We can rewrite $\tau^o_t$ and $\tau^o_e$ as $\tau^o_t=\left({(\ln\frac{1}{1-\sigma})}/({\lambda_a +q(\lambda_b-\lambda_a)})\right)^{{\alpha}/{2}}$ and $\tau^o_e=(q\lambda_c\ln\frac{1}{1-\epsilon})^{-\alpha/2}$.
Clearly, to achieve a positive $\mathbf{\Omega}$, we should ensure $\tau^{o}_{t}>\tau^{o}_{e}$, which is equivalent to
\begin{equation}\label{condition2}
q>q_{m}\triangleq
\left(\Delta-1\right)^{-1}\rho^{-\delta},
\end{equation}
where $\Delta\triangleq \frac{\sigma_o\epsilon_o}{\pi\lambda_eN_er_o^2}$.
This indicates, to meet outage probability constraints, a minimum fraction $q_{m}$ must be guaranteed.
Given that $q_{m}<1$, the choice of $\sigma$ and $\epsilon$ should satisfy
\begin{equation}\label{condition1}
\Delta>1+\rho^{-\delta},
\end{equation}
i.e., too small a $\sigma$ and/or too small an $\epsilon$ might not be promised.
In the following, we only consider the non-trivial case of a positive $\mathbf{\Omega}$, i.e., $q>q_{m}$; and thus, maximizing $\mathbf{\Omega}$ in \eqref{st_hf} is equivalent to maximizing $\ln \frac{1+\tau^o_t}{1+\tau^o_e}$.
Recalling \eqref{beta_t}, we introduce the following auxiliary function
\begin{equation}\label{w}
w(q)=\ln\frac{w_1(q)}{w_2(q)},
\end{equation}
where $w_1(q)=1+\beta_1(1+\rho^{\delta}q)^{-\frac{\alpha}{2}}$
with $\beta_1 \triangleq \left(\frac{\sigma_o}{\kappa\lambda_lr_o^2}\right)^{\frac{\alpha}{2}}$, $w_2(q)=1+\beta_2(\rho^\delta q)^{-\frac{\alpha}{2}}$
with $\beta_2 \triangleq \left(\frac{\pi\lambda_eN_e}{\kappa\lambda_l\epsilon_o}\right)^{\frac{\alpha}{2}}$, and $w_1(q)>w_2(q)>1$ for $q\in(q_{m},1]$.
%Recalling \eqref{beta_t} and \eqref{beta_e}, we express $\tau^o_t$ and $\tau^o_e$ as $\tau^o_t=\beta_1(1+\rho^{\delta}q)^{-\alpha/2}$ and $\tau^o_e=\beta_2q^{-\alpha/2}$ with $\beta_1 \triangleq \left(\left({\ln\frac{1}{1-\sigma}}\right)/\left({\kappa\lambda_lr_o^2}\right)\right)^{\alpha/2}$ and $\beta_2 \triangleq \left(\left({\pi\lambda_eN_e}\right)/\left[{\kappa\lambda_l\rho^{\delta}\ln{\frac{1}{1-\epsilon}}}\right]\right)^{\alpha/2}$.
%Thereby, $w(q)$ can be rewritten as
%\begin{equation}\label{w}
%w(q)=
%\ln\left(\frac{1+\beta_1(1+\rho^{\delta}q)^{-\alpha/2}}
%{1+\beta_2q^{-\alpha/2}}\right)
%\end{equation}
Hence, problem \eqref{st_hf} changes to
\begin{align}\label{f_max_def}
\max_{q}~ w(q),\quad
\mathrm{s.t.}~ 0<q_{m}<q\le 1.
\end{align}
Fortunately, we also successfully prove
the quasi-concavity of $w(q)$ in $q$ and provide the optimal solution of problem \eqref{f_max_def} in the following theorem.
\begin{theorem}\label{opt_q_st_theorem}
	The optimal fraction of FD receivers that maximizes the NST $\mathbf{\Omega}$ in \eqref{st_hf} is
	\begin{equation}\label{opt_q_st}
	q_{st}^* =
	\begin{cases}
	~\emptyset, & {\pi\lambda_eN_e}{\epsilon_o^{-1}}\in[ {X},\infty),\\
	~1,& {\pi\lambda_eN_e}{\epsilon_o^{-1}}\in[ {Y}, {X}),\\
	~q_{st}^{\circ},&  {\pi\lambda_eN_e}{\epsilon_o^{-1}}\in(0, {Y}),
	\end{cases}
	\end{equation}
	where $q_{st}^* = \emptyset$ corresponds to an empty feasible region of $q$,
	$ {X}\triangleq \sigma_o/\left(r_o^2\left(1+\rho^{-\delta}\right)\right) $, $ {Y}\triangleq \left(\kappa^{-\alpha/2}\lambda_l^{-\alpha/2}\rho^{-(1+\delta)}+\left(1+\rho^{-\delta}\right) {X}^{-\alpha/2}\right)^{-\delta}< {X}$ \footnote{$ {Y}< \left(\left(1+\rho^{-\delta}\right) {X}^{-\alpha/2}\right)^{-\delta}=\left(1+\rho^{-\delta}\right)^{-\delta} {X}< {X}$.}, and
	$q_{st}^{\circ}$ is the unique root $q$ that satisfies %$1-\frac{1+\rho^{\delta}q+\beta_1^{-1}\left(1+\rho^{\delta}q\right)^{1+\alpha/2}}
	%{\rho^{\delta}q+\beta_2^{-1}\rho^{\delta}q^{1+\alpha/2}
	%}=0$,
	\begin{equation}\label{opt_q_st_exp}
	1-\frac{1+\rho^{\delta}q+\beta_1^{-1}\left(1+\rho^{\delta}q\right)^{1+\alpha/2}}
	{\rho^{\delta}q+\beta_2^{-1}(\rho^{\delta}q)^{1+\alpha/2}
	}=0.
	\end{equation}
	The LHS of \eqref{opt_q_st_exp} is a monotonically increasing function of $q$ in the range $q\in(q_{m},1]$, and is first negative and then positive when ${\pi\lambda_eN_e}{\epsilon_o^{-1}}\in(0, {Y})$; and thus, the value of $q_{st}^{\circ}$ can be efficiently calculated using the bisection method with \eqref{opt_q_st_exp}.
	
\end{theorem}

\begin{proof}
	Please refer to Appendix \ref{appendix_opt_q_st_theorem}.
\end{proof}

Theorem \ref{opt_q_st_theorem} shows that when $N_e$ or $\lambda_e$ is sufficiently small such that ${\pi\lambda_eN_e}\epsilon_o^{-1}<{X}$, there exists a unique fraction $q$ that maximizes NST $\mathbf{\Omega}$; otherwise no positive $\mathbf{\Omega}$ can be achieved.

In the following corollary, we provide some insights into the optimal $q^{\circ}_{st}$ given in \eqref{opt_q_st}.
\begin{corollary}\label{corollary_q_st}
	The optimal $q^{\circ}_{st}$ given in \eqref{opt_q_st} monotonically increases with $\lambda_e$, $N_e$ and $r_o$, and monotonically decreases with $\sigma$, $\epsilon$, $\rho$ and $\lambda_l$.
\end{corollary}
\begin{proof}
	Please refer to Appendix \ref{appendix_corollary_q_st}.
\end{proof}

Corollary \ref{corollary_q_st} indicates that under a moderate constraint on the connection outage probability (a large $\sigma$) or on the secrecy outage probability (a large
$\epsilon$), we should reduce the portion of FD receivers.
This is because, on one hand, reducing FD receivers decreases interference such that greatly benefits legitimate transmissions  especially when a large $\sigma$ is tolerable; on the other hand, if a large $\epsilon$ is tolerable, we need fewer FD jammers against eavesdropping.
It is worth mentioning that the optimal fraction $q^{\circ}_{st}$ increases as $r_o$ increases, which is just the opposite of what we have observed in Corollary \ref{corollary_q_sl}.
The reason behind is that here we have ignored self-interference and meanwhile eavesdroppers who are close to a legitimate transmitter is less impaired by the paired FD receiver as $r_o$ increases, hence more FD receivers are needed.

The aforementioned theoretic results are validated in Fig. \ref{OPT_Q_ST}, where we see that the optimal fraction $q^{\circ}_{st}$ deeply depends on parameters $r_o$ and $\sigma$.  When $r_o$ is large and meanwhile $\sigma$ is small (e.g., $r_o=2$, $\sigma=0.1$) such that condition \eqref{condition1} is violated, there exists no positive NST no matter how the network design allocates FD and HD receivers.
That means in such a legitimate transmission distance, the connection outage probability requirement is too rigorous to satisfy.
In order to achieve a certain level of NST, network design should have to relax the connection outage probability constraint or shorten the legitimate distance.

\begin{figure}[!t]
	\centering
	\includegraphics[width=3.0in]{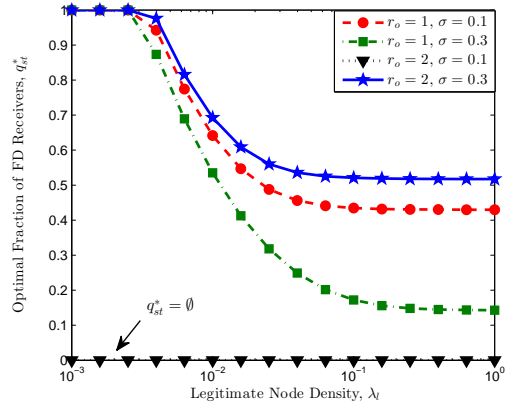}
	\caption{The optimal fraction of FD receivers that maximizes NST $\mathbf{\Omega}$ vs. $\lambda_l$ for different values of $r_o$ and $\sigma$, with $\alpha=4$, $\lambda_e = 10^{-4}$, $N_e=4$, $\rho=2$ and $\epsilon=0.05$.}
	\label{OPT_Q_ST}
\end{figure}

Let us recall \eqref{phi_rho} in Appendix \ref{appendix_corollary_q_st}, it is not difficult to deduce that $q^{\circ}_{st}$ is inversely proportional to $\rho^\delta$, since $\rho^\delta q^{\circ}_{st}$ keeps constant in $\phi(\rho^\delta q^{\circ}_{st})=0$ when the other parameters are fixed.
As a consequence, $q^{\circ}_{st}\rightarrow 0$ as $\rho\rightarrow \infty$.
If we consider a dense network by letting $\lambda_l\rightarrow \infty$ in \eqref{phi_rho_2}, we can further obtain a simple expression for $q_{st}^{\circ}$ given below
\begin{equation}
q_{st}^{\circ,\lambda_l\rightarrow\infty}=\left({\Delta^{{1}/{(1+\delta)}}-1}\right)^{-1}\rho^{-\delta},
\end{equation}
which is independent of $\lambda_l$, just as shown in Fig. \ref{OPT_Q_ST}.
This is different from what we can see in Fig. \ref{OPT_Q_SL} where the optimal $q_{sl}^*$ goes to zero as $\lambda_l$ goes to infinity. This is because a positive secrecy rate mainly depends on the relative interference strength between  legitimate nodes and eavesdroppers as $\lambda_l$ goes to infinity, and a certain portion of FD receivers must be activated to ensure the superiority of the main channel over the wiretap channel in terms of channel quality.
%of the interference to the legitimate node constraints of the connection outage probability and the secrecy outage probability here.

Substituting $q^*_{st}$ given in \eqref{opt_q_st} into \eqref{st_hf}, we obtain the maximum NST $\mathbf{\Omega}^*$.
Fig. \ref{NWST_COMPARE} compares this NST $\mathbf{\Omega}^*$ and those obtained at fixed $q$'s.
%We see that $\mathbf{\Omega}^*$ becomes small under a more rigorous secrecy outage probability constraint, i.e., a small value of $\epsilon$.
Obviously, activating a proper fraction of FD receivers significantly improves NST.
For example, the optimal fraction $q_{st}^*$ increases NST by about 28$\%$ than the equal proportion case (i.e., $q=0.5$), and by up to 1400$\%$ than the small-$q$ case (e.g., $q=0.1$).
We can observe that, too small a  $\sigma$ might not be satisfied while too large a $\sigma$ results in a small successful transmission probability and accordingly small NST.
Therefore, a moderate constraint on the connection outage probability is desirable for improving NST.
Fig. \ref{NWST_COMPARE} also illustrates the influence of the path loss exponent $\alpha$ on NST.
A general trend is that NST increases as $\alpha$ becomes larger.
The reason behind is that the distance between a legitimate transmitter-receiver pair is small such that the signal attenuation in a legitimate link is less significant than it is in the eavesdropper link.
This implies short-range secure communications might prefer a large path-loss exponent, especially in a sparse-eavesdropper environment.

\begin{figure}[!t]
	\centering
	\includegraphics[width=3.0in]{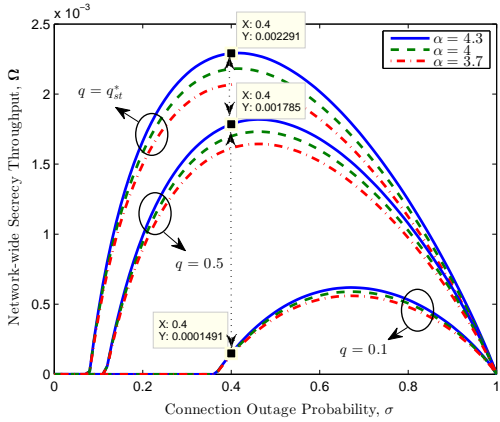}
	\caption{NST vs. $\sigma$ for different values of and $\alpha$, with $\lambda_l=10^{-3}$, $\lambda_e = 10^{-4}$, $N_e=4$, $r_o=1$, $\epsilon=0.01$ and $\rho=1$.}
	\label{NWST_COMPARE}
\end{figure}

\section{Network-wide Secrecy Energy Efficiency}
In this section, we determine the optimal fraction $q$ of FD receivers that maximizes the NSEE $\mathbf{\Psi}$ with and without considering a required minimum NST.
As presented in previous sections, we consider the scenario where self-interference is efficiently canceled.

\subsection{Without NST Constraint}
In this subsection, we ignore the requirement of a minimum NST.
Substituting \eqref{st_hf} into \eqref{ee_def}, the optimization  problem of interest can be formulated as
%\begin{subequations}
\begin{align}\label{ee_max_def}
&\max_{q}~ \mathbf{\Psi}=\frac{\lambda_l(1-\sigma) w(q) }{\lambda_l(P_t+P_c)+q\lambda_lP_j},~~\mathrm{s.t.}~~ 0<q_{m}<q\le 1,
\end{align}
%\end{subequations}
where $w(q)$ and $q_{m}$ have been defined in \eqref{w} and \eqref{f_max_def}, respectively.
Introducing $\rho_c\triangleq \frac{P_j}{P_t+P_c}$ and
the following auxiliary function
\begin{equation}\label{J}
J(q) = \frac{w(q)}{1+\rho_c q},
\end{equation}
the object function $\mathbf{\Psi}$ of problem \eqref{ee_max_def} can be rewritten in the form of $\mathbf{\Psi}=\frac{1-\sigma}{P_t+P_c}J(q)$.
Clearly, maximizing $\mathbf{\Psi}$ is equivalent to maximizing $J(q)$.
In the following theorem, we prove
the quasi-concavity of $J(q)$ in $q$, again, and provide the optimal solution to problem \eqref{ee_max_def}.

\begin{theorem}\label{opt_q_ee_theorem}
	The optimal fraction of FD receivers that maximizes the NSEE $\mathbf{\Psi}$ in \eqref{ee_max_def} is
	\begin{equation}\label{opt_q_ee}
	q_{ee}^* =
	\begin{cases}
	~\emptyset, & {\pi\lambda_eN_e}{\epsilon_o^{-1}}\ge X,\\
	~1,& {\pi\lambda_eN_e}{\epsilon_o^{-1}}<{X}~\&~\frac{W}{w(1)}>\frac{\delta\rho_c}{1+\rho_c},\\
	~q_{ee}^{\circ},&  \text{otherwise},
	\end{cases}
	\end{equation}
	where $W\triangleq 1-{w_2^{-1}(1)}-\frac{1-w_1^{-1}(1)}{1+\rho^{-\delta} }$ and $X$ has been defined in Theorem \ref{opt_q_st_theorem}.
	In \eqref{opt_q_ee}, $q_{ee}^{\circ}$ is the unique root $q$ of the following equation
	\begin{equation}\label{opt_q_ee_exp}
	Q(q)=0,
	\end{equation}
	where $Q(q)=w^{'}(q)(1+\rho_cq)-\rho_cw(q)$ is initially positive and then negative as $q$ increases; and thus the value of $q_{ee}^{\circ}$ can be efficiently computed using the bisection method with \eqref{opt_q_ee_exp}.
\end{theorem}
\begin{proof}
	Please refer to Appendix \ref{appendix_opt_q_ee_theorem}.
\end{proof}

In the following corollary, we develop some insights into the behavior of $q^{\circ}_{ee}$ given in \eqref{opt_q_ee}.
\begin{corollary}\label{corollary_q_ee}
	The optimal $q^{\circ}_{ee}$ given in \eqref{opt_q_ee} monotonically increases with $\lambda_e$, $N_e$ and $r_o$, and monotonically decreases with $\sigma$, $\epsilon$ and $\rho$.
\end{corollary}
\begin{proof}
	Please refer to Appendix \ref{appendix_corollary_q_ee}.
\end{proof}

The properties of $q^{\circ}_{ee}$ follow Corollary \ref{corollary_q_st}.
Fig. \ref{OPT_Q_EE} depicts the optimal fraction $q^{*}_{ee}$ and verifies Corollary \ref{corollary_q_ee} well.
We see that, too small $\sigma$ and $\epsilon$ might not be simultaneously satisfied (e.g., $\sigma=0.1$, $\epsilon=0.01$).
As can be observed, the optimal $q^{*}_{ee}$ keeps large in the small $\rho$ region, and dramatically decreases as $\rho$ increases.
This is because the increase of jamming power provides a relief to the need of FD jammers.
In addition, as $\sigma$ or $\epsilon$ decreases, the feasible region of $\rho$ that produces a positive $\mathbf{\Psi}$ reduces.
This suggests that to meet
more rigorous connection outage and secrecy outage constraints, we should consume more power in sending jamming signals.
\begin{figure}[!t]
	\centering
	\includegraphics[width=3.0in]{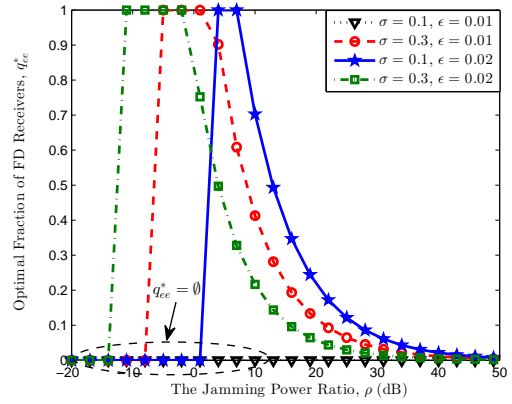}
	\caption{The optimal fraction of FD receivers that maximizes NSEE $\mathbf{\Psi}$ vs. $\rho$ for different values of $\sigma$ and $\epsilon$, with $\alpha=4$, $\lambda_l= 10^{-3}$, $\lambda_e = 10^{-4}$, $N_e=4$ and $r_o=1$.}
	\label{OPT_Q_EE}
\end{figure}

Fig. \ref{NSEE_COMPARE} depicts NSEE versus $\rho$ for different values of $q$.
We see that as $\rho$ increases, NSEE first increases and then decreases.
The underlying reason is that too small jamming power makes NST small whereas too large jamming power leads to large power consumption; both aspects result in small NSEE.
We also find that adaptively adjusting the fraction of FD receivers to jamming power  significantly improves NSEE compared with fixed-$q$ cases, although the latter can approach the optimal performance in some specific regions, e.g., $q=1$ in the small $\rho$ region.

In the following corollary, we reveal how the legitimate node density $\lambda_l$ influences the optimal allocation between FD and HD receivers and the corresponding NSEE.
\begin{corollary}\label{sparse_case}
	In a sparse network, i.e., $\lambda_l\rightarrow 0$, both the optimal fraction $q^*_{ee}$ and the maximum NSEE $
	\mathbf{\Psi}^*$ keep constant, which are independent of $\lambda_l$.
\end{corollary}
\begin{proof}
	Please refer to Appendix \ref{appendix_sparse_case}.
\end{proof}

\begin{figure}[!t]
	\centering
	\includegraphics[width=3.0in]{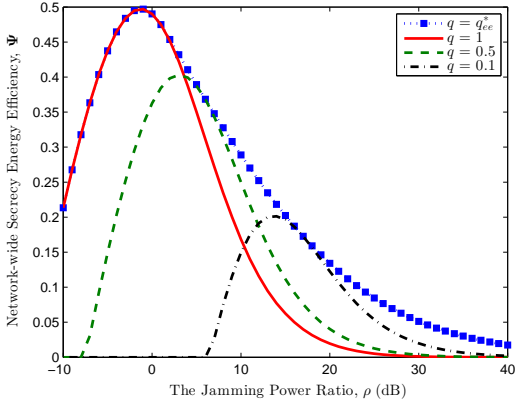}
	\caption{NSEE $\mathbf{\Psi}$ vs. $\rho$ for different values of $q$, with $\alpha=4$, $\lambda_l= 10^{-3}$, $\lambda_e = 10^{-4}$, $N_e=4$, $\sigma=0.3$, $\epsilon=0.02$ and $r_o=1$.}
	\label{NSEE_COMPARE}
\end{figure}

\subsection{With NST Constraint}
For more practical design, we should also take NST into consideration when maximizing NSEE.
In this subsection, we impose a constraint on problem \eqref{ee_max_def} that NST $\mathbf{\Omega}$ lies above threshold ${\Omega}^{\circ}$, i.e., $\mathbf{\Omega}>{\Omega}^{\circ}$.
Since we have already obtained the maximum $\mathbf{\Omega}^*$ in Sec. VI-A, for convenience, we only consider the case $\mathbf{\Omega}^*> {\Omega}^{\circ}$ here.
If $\mathbf{\Omega}^*\le {\Omega}^{\circ}$, we just set $\mathbf{\Psi}$ to zero.
\begin{corollary}\label{with_st_constraint}
	The optimal fraction of FD receivers that maximizes the NSEE $\mathbf{\Psi}$ in \eqref{ee_max_def} subject to the constraint $\mathbf{\Omega}>{\Omega}^{\circ}$ is given as follows
	\begin{align}\label{opt_q_ee_st}
&	q_{ee}^{\star} =
	\nonumber\\
&	\begin{cases}
	\emptyset, & {\pi\lambda_eN_e}{\epsilon_o^{-1}}\ge X,\\
	1,& {\pi\lambda_eN_e}{\epsilon_o^{-1}}<{X}~\&~\frac{W}{w(1)}>\frac{\delta\rho_c}{1+\rho_c},\\
	\max\left(q_{st}^{(1)},q^{\circ}_{ee}\right),& Y\le{\pi\lambda_eN_e}{\epsilon_o^{-1}}<{X}~\&~\frac{W}{w(1)}\le\frac{\delta\rho_c}{1+\rho_c},\\
	q_{ee}^{\tiny{+}},&  \text{otherwise},
	\end{cases}
	\end{align}
	where $q_{ee}^{\circ}$ has bee given in \eqref{opt_q_ee} and $q_{ee}^{+}$ is determined as
	\begin{equation}\label{opt_q_ee_add}
	q_{ee}^{+}=
	\begin{cases}
	~q_{ee}^{\circ}, & q_{st}^{(1)}\le q_{ee}^{\circ}<q_{st}^{(2)},\\
	~q_{st}^{(1)}, & q_{ee}^{\circ}<q_{st}^{(1)},\\
	~q_{st}^{(2)}, & q_{ee}^{\circ}\ge q_{st}^{(2)}.
	\end{cases}
	\end{equation}
	Let us denote $\mathbf{\Omega}(q)$ as a function of $q$.
	If there exists only one root $q\in(q_m,1]$ that satisfies $\mathbf{\Omega}(q)={\Omega}^{\circ}$, we denote this root as $q_{st}^{(1)}$; if there are two such roots, we denote them as $q_{st}^{(1)}$ and $q_{st}^{(2)}$ such that $q_{st}^{(1)}<q_{st}^{(2)}$.
\end{corollary}
\begin{proof}
	Please refer to Appendix \ref{appendix_with_st_constraint}.
\end{proof}

Fig. \ref{Max_EE} shows the maximum NSEE $\mathbf{\Psi}^*$ with $q_{ee}^{*}$ in \eqref{opt_q_ee} and $\mathbf{\Psi}^{\star}$ with $q_{ee}^{\star}$ in \eqref{opt_q_ee_st}.
As indicated in Corollary \ref{corollary_q_ee}, $\mathbf{\Psi}^*$ keeps constant in the small $\lambda_l$ region, whereas $\mathbf{\Psi}^{\star}$
becomes zero since the constraint $\mathbf{\Omega}>{\Omega}^{\circ}$ is not satisfied.
When $\lambda_l$ falls in the medium range, the curve of $\mathbf{\Psi}^*$ and its counterpart $\mathbf{\Psi}^{\star}$ merge and vary smoothly.
As $\lambda_l$ increases further, both $\mathbf{\Psi}^*$ and $\mathbf{\Psi}^{\star}$ quickly drop to zero.
Therefore, a moderate network density is desirable.
Fig. \ref{Max_EE} also indicates that although increasing jamming power helps to suppress eavesdroppers, it is at a cost of energy efficiency.

\begin{figure}[!t]
	\centering
	\includegraphics[width=3.0in]{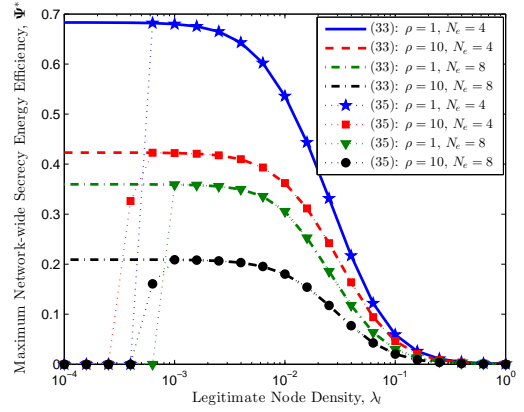}
	\caption{The maximum NSEE $\mathbf{\Psi}^*$ vs. $\lambda_l$ for different values of $\rho$ and $N_e$, with $\alpha=4$, $\lambda_e = 10^{-4}$, $r_o=1$, $\sigma=0.3$, $\epsilon=0.03$ and ${\Omega}^{\circ}=0.001$.
		The results labeled \eqref{opt_q_ee} and \eqref{opt_q_ee_st} are obtained without and with the constraint $\mathbf{\Omega}>{\Omega}^{\circ}$, respectively.}
	\label{Max_EE}
\end{figure}

To better guide network designers on how to well design the network, Table \ref{SUM_OPTIMAL} summarizes the relationships between the optimal fraction $q$ of FD receivers and key parameters in different objectives.
	
	\begin{table}[t]
		\caption{Relationships between Optimal $q^*$ and Key Parameters}
		\begin{center}
			\begin{tabular}{|c|c|c|}
				\hline
				Objectives&  $q^*$ increases with& $q^*$ decreases with\\
				\hline
				ASLN &  $\lambda_e$, $N_e$ & $\lambda_l$, $r_o$, $\eta$, $\rho$, $\tau_t$, $\tau_e$ \\
				\hline
				NST &  $\lambda_e$, $N_e$, $r_o$ & $\lambda_l$, $\rho$, $\sigma$, $\epsilon$ \\
				\hline
				NSEE &  $\lambda_e$, $N_e$, $r_o$ & $\rho$, $\sigma$, $\epsilon$ \\
				\hline
			\end{tabular}
		\end{center}
		\label{SUM_OPTIMAL}
	\end{table}

\section{Conclusion and Future Work}
In this paper, we study physical layer security in a wireless ad hoc network with a hybrid full-/half-duplex receiver deployment strategy.
We provide a comprehensive performance analysis and network design under a stochastic geometry framework.
We first analyze connection outage and secrecy outage probabilities for a typical legitimate link, and show that enabling more FD receivers increases the connection outage probability but decreases the secrecy outage probability.
Based on the analytical results of the dual probabilities, we prove that ASLN, NST and NSEE are all \emph{quasi-concave} on the fraction of FD receivers, and maximize each of them by providing the optimal fraction.
We further develop various useful properties on this optimal fraction.
Numerical results are demonstrated to validate our theoretical findings.

This paper opens up several interesting research directions. For example, the proposed framework can be extended to investigate the cooperative or multi-antenna FD receivers, where additional degrees of freedom might be gained not only in alleviating the self-interference but also in designing the jamming signals.
	The benefit of FD receiver jamming techniques can be further exploited by jointly optimizing the allocation between FD and HD receivers and the jamming transmit power of each FD receiver, given that the latter also strikes a non-trivial tradeoff between reliability and secrecy. Another possible direction for future research is to consider the randomness of self-interference and propose an adaptive and intelligent criterion to select work mode for receivers, e.g., letting those receivers with instantaneous self-interference power lying below a certain value work in the FD mode and the rest work in the HD mode.

\appendix

\subsection{Proof of Theorem \ref{opt_q_sl_theorem}}
\label{appendix_opt_q_sl_theorem}
We start by taking the first-oder derivative of $ {F}(q)$ on $q$
\begin{equation}\label{dF1}
{{F}^{'}(q)} =  {K}(q)e^{-{B}q-{C}/q},
\end{equation}
where
\begin{equation}\label{K}
{K}(q) = \left(A+{1}/{q}-1\right)\left(1+{C}/{q}-Bq\right)-{1}/{q}.
\end{equation}
To determine the sign of
${{F}^{'}(q)}$, we first investigate the behavior of $K(q)$ at the boundaries $q\rightarrow 0^+$ and $q=1$, respectively.
Substituting $q\rightarrow 0^+$ into \eqref{K} yields
$\lim_{q\rightarrow 0^+} {K}(q)
=\lim_{q\rightarrow 0^+}\left({C}/{q^2}\right)>0$.
Substituting $q=1$ into \eqref{K} yields $ {K}(1)=A(1+C-B)-1$, the sign of which relies on specific values of $A$, $B$ and $C$.
Consider the following two cases.
%We divide the discussion of the monotonicity of $ {F}(q)$ in $q$ into the following two cases.

1) ${K}(1)> 0$:
We have $A(1+C-B)>  1\Rightarrow C-B>{1}/{A}-1>0$, which yields $C/q-Bq\ge C-B>0$.
Substituting this inequality along with $1/q>1$ into \eqref{dF1}, we obtain
$ {K}(q)>  A(1+C-B)-1> 0$, i.e., ${{F}^{'}(q)}>0$.
This means $ {F}(q)$ monotonically increases in $q$ within the entire range $q\in(0,1]$, and the optimal $q$ that maximizes $ {F}(q)$ or $\bm{N}$ is $q^*=1$.

2) {${K}(1)< 0$}:
There at least exists one point $q\in(0,1]$ that satisfies ${K}(q)=0$ since $K(q)$ is a continuous function of $q$ and
$\lim_{q\rightarrow 0^+} {K}(q)>0$.
Denote an arbitrary zero-crossing point $q$ of
$K(q)$ as $q_{o}$, i.e., ${K}(q_o)=0$.
To determine the monotonicity of $ {F}(q)$ in $q$, we first take the second-order derivative of ${F}(q)$ at $q=q_o$ from \eqref{dF1}
\begin{equation}\label{dF2}
{F}^{''}(q_o) =  {K}'(q_o)e^{-{B}q_o-{C}/q_o},
\end{equation}
where
\begin{equation}\label{dK1}
{K}'(q_o)=
\left(B+{C}{q_o^{-2}}\right)A+
{2C}{q_o^{-3}}-{C}{q_o^{-2}}-B.
\end{equation}
Clearly, the sign of $ {F}^{''}(q_o)$ follows that of $ {K}'(q_o)$.
We resort to the equation $ {K}(q_o)=0$ in \eqref{K}, which yields $A = 1-\frac{1}{q_o}\left(1-\frac{1}{1+{C}/{q_o}-Bq_o}\right)$.
Given that $0<A<1$, we readily obtain $C/q>Bq$, substituting which combined with $0<q_o\le 1$ into \eqref{dK1} yields $ {K}'(q_o)<0$, i.e., $ {F}''(q_o)<0$.
Invoking the definition of single-variable quasi-concave
function \cite[Sec. 3.4.2]{Boyd2004Convex},
we conclude that $ {F}(q)$ is a quasi-concave function of $q$, and there exists a unique $q$ that maximizes $ {F}(q)$.
In other words, $ {F}(q)$ initially increases and then decreases in $q$, and the peak value of $ {F}(q)$ is achieved at the unique root $q$ of the equation ${K}(q)=0$.
By now, we have completed the proof.

\subsection{Proof of Corollary \ref{corollary_q_sl}}
\label{appendix_corollary_q_sl}

Recall \eqref{K}, and the optimal $q^{\circ}_{sl}$ satisfies $K(q^{\circ}_{sl})=0$.
We first take the first-order derivative of $q^{\circ}_{sl}$ on $A$ using the derivative rule for implicit functions
with $K(q^{\circ}_{sl})=0$, i.e.,
\begin{equation}
\frac{d q^{\circ}_{sl}}{d A}
=-\frac{\partial K(q^{\circ}_{sl})/\partial A}
{\partial K(q^{\circ}_{sl}) /\partial q^{\circ}_{sl}}.
\end{equation}
From \eqref{K}, we have
${\partial K(q^{\circ}_{sl}) /\partial A}=1+C/q^{\circ}_{sl}-Bq^{\circ}_{sl}>0$.
From \eqref{dK1}, we know that ${\partial K(q^{\circ}_{sl}) /\partial q^{\circ}_{sl}}<0$.
Thus, we obtain ${d q^{\circ}_{sl}}/{d A}>0$.
In a similar way, we can prove ${d q^{\circ}_{sl}}/{d B}<0$ and ${d q^{\circ}_{sl}}/{d C}>0$.
Observing the expressions of $A$, $B$ and $C$ directly yields the relationships between the optimal $q_{sl}^{\circ}$ and the relevant parameters.
For the perfect SIC case, substituting $\eta=0$, or, equivalently, $A=1$, into \eqref{opt_sl_q_exp} directly yields the result given in \eqref{q_sl_sic}.

\subsection{Proof of Theorem \ref{opt_q_st_theorem}}
\label{appendix_opt_q_st_theorem}

Taking the first-order derivative of $w(q)$ in \eqref{w} on $q$, i.e.,
\begin{equation}\label{dw_1}
w^{'}(q)
=\frac{w^{'}_1(q)}{w_1(q)}- \frac{w_2^{'}(q)}{w_2(q)},
\end{equation}
where $w^{'}_1(q)=-\frac{{\rho^{\delta}}[w_1(q)-1]}{{\delta\left(1+\rho^{\delta}q\right) }}$ and $w^{'}_2(q)=
-\frac{w_2(q)-1}{{\delta}q}$.
% are, respectively, given by
%\begin{align}
%\label{dw1_1}
%w^{'}_1(q)&
%=-\frac{{\alpha}\rho^{\delta}}{2(1+\rho^{\delta}q )}\left(w_1(q)-1\right),\\
%\label{dw2_1}
%w^{'}_2(q)&=
%-\frac{{\alpha}
%}{{2}q}(w_2(q)-1).
%\end{align}
Directly determining either the sign of $w^{'}(q)$ or the concavity of $w(q)$ from the second-order derivative $w^{''}(q)$ is difficult.
Instead, we prove the quasi-concavity of $w(q)$ on $q$ by reforming $w^{'}(q)$ as $w^{'}(q)=
\frac{w^{'}_1(q)}{w_1(q)}\phi(q)$
%\begin{equation}\label{dw_11}
%w^{'}(q)=
%\frac{w^{'}_1(q)}{w_2(q)}\phi(q)
%\end{equation}
with $\phi(q)$ given by
\begin{align}\label{phi}
\phi(q)=1-\frac{w_1(q)w_2^{'}(q)}{w^{'}_1(q)w_2(q)}=1-\frac{\left(1+\rho^{\delta}q\right)w_1(q)(w_2(q)-1)}{\rho^{\delta}q (w_1(q)-1) w_2(q)}.
\end{align}
Apparently, the first term ${w^{'}_1(q)}/{w_1(q)}$
is negative.
Next, we determine the sign of $\phi(q)$.
Taking the first-order derivative of $\phi(q)$ on $q$, and after some algebraic manipulations, we obtain
\begin{align}\label{dphi_1}
\phi^{'}(q)=\frac{\left[1+\delta w_2(q)\right]w_1(q)+\rho^{\delta}q\left[w_1(q)-w_2(q)\right]}{\delta \rho^{\delta} q^2 [w_1(q)-1]w_2^2(q)/[w_2(q)-1]}.
\end{align}
Since $w_1(q)>w_2(q)>1$, $\phi^{'}(q)>0$ always holds, i.e., $\phi(q)$ monotonically increases with $q$.
When $q=q_{m}$, we have $w_1(q_{m})=w_2(q_{m})$ and $\frac{w^{'}_2(q_{m})}{w^{'}_1(q_{m})}=\frac{1+\rho^{\delta}q_{m}}{\rho^{\delta}q_{m}}$, thus
$\phi(q_{m})=-\frac{1}{\rho^{\delta}q_{m}}<0$.
When $q = 1$, $\phi(1)=1-\frac{\beta_2(1+\rho^{\delta})\left(\beta_1+(1+\rho^{\delta})^{\alpha/2}\right)}{\beta_1\rho^{\delta}(\beta_2+\rho)}$,
%\begin{equation}\label{q_1}
%\phi(1)=1-\frac{\beta_2(1+\rho^{\delta})\left(\beta_1+(1+\rho^{\delta})^{\alpha/2}\right)}{(1+\beta_2)\beta_1\rho^{\delta}}.
%\end{equation}
the sign of which depends on $\beta_1$ and $\beta_2$.
Specifically, if $1+\beta_1^{-1}(1+\rho^{\delta})^{1+\alpha/2}>\beta_2^{-1}\rho^{1+\delta}$,
%i.e., $1+\left(1+\rho^{\delta}\right)^{1+\alpha/2}\left((\kappa\lambda_lr^2)/\left(\ln\frac{1}{1-\sigma}\right)\right)^{\alpha/2}
%>
%\rho^{\delta}\left(\left({\kappa\lambda_l\rho^{\delta}\ln{\frac{1}{1-\epsilon}}}\right)/\left({\pi\lambda_eN_e}\right)\right)^{\alpha/2}$,
we have $\phi(1)<0$; otherwise, $\phi(1)\ge 0$.
In the following, we derive the optimal $q$ that maximizes $w(q)$ by distinguishing two cases.

1) If $\phi(1)<0$, $\phi(q)<0$ holds in the entire range $q\in(q_{m},1]$.
Accordingly, we have $w^{'}(q)>0$, i.e., $w(q)$ monotonically increases with $q$.
Therefore, the optimal $q$ that maximizes $w(q)$ is $q^*=1$.

2) If $\phi(1)\ge 0$, $\phi(q)$ is initially negative and then positive in the range $q\in(q_{m},1]$; the zero-crossing point $q$ that satisfies $\phi(q)=0$ is denoted by $q_o$.
We can also conclude that $w^{'}(q)$ is initially positive and then negative after $q$ exceeds $q_o$.
In other words, $w(q)$ first increases and then decreases with $q$, and $q_o$ is the solution that yields the peak value of $w(q)$.

By now, we have proved the quasi-concavity of $w(q)$ on $q$.
Combined with %$q_{m}<1\Rightarrow ({\pi\lambda_eN_e})/\left({\ln\frac{1}{1-\epsilon}}\right) < \left(\ln\frac{1}{1-\sigma}\right)/\left(r^2\left(1+\rho^{-\delta}\right)\right)$
$q_{m}<1\Rightarrow \Delta>1+\rho^{-\delta}$
and the given results, we complete the proof.

\subsection{Proof of Corollary \ref{corollary_q_st}}
\label{appendix_corollary_q_st}

Recall \eqref{phi}, and the optimal $q^{\circ}_{st}$ satisfies $\phi(q^{\circ}_{st})=0$.
Taking the first-order derivative of $q^{\circ}_{st}$ on $\beta_1$ using the derivative rule for implicit functions
with $\phi(q^{\circ}_{st})=0$ yields
\begin{equation}
\frac{d q^{\circ}_{st}}{d \beta_1}
=-\frac{\partial \phi(q^{\circ}_{st})/\partial \beta_1}
{\partial \phi(q^{\circ}_{st}) /\partial q^{\circ}_{st}},
\end{equation}
where $\frac{\partial \phi(q^{\circ}_{st})}{\partial \beta_1}=\frac{{\beta_1^{-2}\left(1+\rho^{\delta}q\right)^{1+\alpha/2}}}{	{\rho^{\delta}q+\beta_2^{-1}(\rho^{\delta}q)^{1+\alpha/2}}}>0$ and $\frac{\partial \phi(q^{\circ}_{st})}{\partial q^{\circ}_{st}}>0$ (see \eqref{dphi_1}); and thus, $\frac{d q^{\circ}_{st}}{d \beta_1}<0$.
Similarly, we can prove $\frac{d q^{\circ}_{st}}{d \beta_2}>0$.
From the expressions of $\beta_1$ and $\beta_2$ given in \eqref{w}, we can infer that $q^{\circ}_{st}$ increases in $\lambda_e$, $N_e$ and $r_o$, while decreases in $\sigma$ and $\epsilon$.
As to $\frac{dq^{\circ}_{st}}{d\rho}$, we first express $\phi(\rho^\delta q)$ as
% let $b_1\triangleq \beta_1^{-1}=\left(\frac{\kappa\lambda_lr_o^2}{\ln\frac{1}{1-\sigma}}\right)^{\alpha/2}$ and $b_2\triangleq 1/(\beta_2\rho)
%=\left(\frac{\kappa\lambda_l\ln{\frac{1}{1-\epsilon}}}{\pi\lambda_eN_e}\right)^{\alpha/2}$ 	 such that
\begin{equation}\label{phi_rho}
\phi(\rho^\delta q)= 1 - \frac{1+\rho^\delta q+\beta_1^{-1}(1+\rho^\delta q)^{1+\alpha/2}}{\rho^\delta q+\beta_2^{-1}(\rho^\delta q)^{1+\alpha/2}}.
\end{equation}
Taking the first-order derivative of $\phi(\rho^\delta q)$ on $\rho^\delta q$ and invoking the equation $\phi(\rho^\delta q^{\circ}_{st})=0$, we can prove $\frac{\partial\phi(q^{\circ}_{st})}{\partial \rho}>0$.
Thereby, we have $\frac{d q^{\circ}_{st}}{d \rho}
=-\frac{\partial \phi(q^{\circ}_{st})/\partial \rho}
{\partial \phi(q^{\circ}_{st}) /\partial q^{\circ}_{st}}<0$.
As to $\frac{dq^{\circ}_{st}}{d\lambda_l}$, we let
\begin{equation}\label{phi_rho_2}
\phi(\lambda_l)= 1 - \frac{1+\rho^\delta q+b_1\lambda_l^{\alpha/2}(1+\rho^\delta q)^{1+\alpha/2}}{\rho^\delta q+b_2\lambda_l^{\alpha/2}(\rho^\delta q)^{1+\alpha/2}},
\end{equation}
%$b_1\triangleq \beta_1^{-1}=\left(\frac{\kappa\lambda_lr_o^2}{\ln\frac{1}{1-\sigma}}\right)^{\alpha/2}$ and $b_2\triangleq 1/(\beta_2\rho)
%=\left(\frac{\kappa\lambda_l\ln{\frac{1}{1-\epsilon}}}{\pi\lambda_eN_e}\right)^{\alpha/2}$
where $b_1\triangleq \left(\frac{\kappa r_o^2}{\sigma_o}\right)^{\frac{\alpha}{2}}$ and $b_2\triangleq\left(\frac{\kappa \epsilon_o}{\pi\lambda_eN_e}\right)^{\frac{\alpha}{2}}$.
Taking the first-order derivative of $\phi(\lambda_l)$ on $\lambda_l$ and invoking $\beta_1(1+\rho^{\delta}q)^{-\alpha/2}>\beta_2(\rho^\delta q)^{-\alpha/2}$ in \eqref{w}, we can prove $\frac{\partial\phi(q^{\circ}_{st})}{\partial \lambda_l}>0$ and $\frac{d q^{\circ}_{st}}{d \lambda_l}
=-\frac{\partial \phi(q^{\circ}_{st})/\partial \lambda_l}
{\partial \phi(q^{\circ}_{st}) /\partial q^{\circ}_{st}}<0$.
By now, the proof is complete.

\subsection{Proof of Theorem \ref{opt_q_ee_theorem}}
\label{appendix_opt_q_ee_theorem}

We start by giving the first-order derivative of $J(q)$ on $q$,
\begin{equation}\label{dJ1}
{J^{'}(q)} = {Q(q)}{(1+\rho_cq)^{-2}},
\end{equation}
where $Q(q)$ is given in \eqref{opt_q_ee_exp}.
To proceed, we first give the following lemma which is very important to subsequent proof.
\begin{lemma}\label{lemma_dQ}
	If $w^{'}(q)>0$ for $q\in(q_m,1]$, $Q^{'}(q)<0$ holds.
\end{lemma}
\begin{proof}
	Since $Q^{'}(q)=w{''}(q)(1+Bq)$, to prove $Q^{'}(q)>0$ we need only to prove $w^{''}(q)>0$.
	Taking the derivative of $w^{'}(q)$ in \eqref{dw_1} on $q$ yields $w^{''}(q)$ which is given in \eqref{dw2} at the top of the next page, where the last equality holds for $w_1^{''}(q)=\left({(1+\delta)
		[w_1(q)-1]\rho^{2\delta}}\right)/\left({\delta^2(1+\rho^\delta q)^2}\right)$ and $w_2^{''}(q)=\left({(1+\delta)
		[w_2(q)-1]}\right)/\left({\delta^2q^2}\right)$.
	\begin{figure*}[t]
		\begin{align}\label{dw2}
		w^{''}(q)
		&=\frac{w_1^{''}(q)w_1(q)-(w_1^{'}(q))^2}{w_1^2(q)}
		-\frac{w_2^{''}(q)w_2(q)-(w_2^{'}(q))^2}{w_2^2(q)}\nonumber\\
		&=\frac{(w_1(q)-w_2(q))}{\delta^2 w_1^2(q)w_2^2(q)q^2(1+\rho^\delta q)^2}\Bigg\{
		{\rho^\delta q}\left(\rho^\delta q\left[w_1(q)+w_2(q)-w_1(q)w_2(q)\right]-\frac{2w_1^2(q)[w_2(q)-1]}{w_1(q)-w_2(q)}\right)-\frac{w_1^2(q)[w_2(q)-1]}{ [w_1(q)-w_2(q)]}\nonumber\\
		&\qquad\qquad \qquad\qquad\qquad\qquad +\delta\rho^\delta qw_1(q)w_2(q)\left(\rho^\delta q-\frac{2w_1(q)[w_2(q)-1]}{w_1(q)-w_2(q)}\right)-
		\frac{\delta w_1^2(q)w_2(q)[w_2(q)-1]}{ w_1(q)-w_2(q)}\Bigg\},
		\end{align}
		\hrulefill
	\end{figure*}
	Invoking \eqref{dw_1}, we can readily obtain the following relationship
	\begin{equation}
	w^{'}(q)>0\Rightarrow \frac{w_2(q)-1}{w_1(q)-w_2(q)}>\frac{\rho^\delta q}{w_1(q)}.
	\end{equation}
	Plugging the above inequality into \eqref{dw2} yields
	\begin{align}
	&w^{''}(q) <
	-\frac{w_1(q)-w_2(q)}{\delta^2 w_1(q)w_2^2(q)\rho^{-\delta} q(1+\rho^\delta q)^2}\times\nonumber\\
	&\left({1}+\rho^\delta q[1-w_2(q)/w_1(q)]+\delta(1+\rho^\delta q)w_2(q)\right)<0,
	\end{align}
	which completes the proof.
\end{proof}

Next, we are going to determine the sign of $Q(q)$ or $J^{'}(q)$ in \eqref{dJ1}.
We first determine the sign of $Q(q)$ at the boundaries $q_{m}$ and $1$.
Combined with $w_1(q_m)=w_2(q_m)$, we have
\begin{equation}\label{F_m}
Q(q_{m})=\frac{w_1(q_{m})-1}{\delta w_1(q_{m})}\left(\frac{1}{q_{m}}-\frac{\rho^\delta}{1+\rho^\delta q_{m}}\right)>0.
\end{equation}
Substituting $q=1$ into $Q(q)$ yields
\begin{align}\label{Q_1}
&Q(1)
=w^{'}(1)(1+\rho_c)-\rho_cw(1)=\nonumber\\
&\frac{\alpha}{2}{(1+\rho_c)}\left({1-w_2^{-1}(1)}-\frac{1-w_1^{-1}(1) }{ 1+\rho^{-\delta} }\right)-{\rho_cw(1)},
\end{align}
The sign of $Q(1)$ depends on the values of involved parameters.
Let us distinguish two cases.

1) If $Q(1)> 0$, we have $w^{'}(1)>0$.
Since $w(q)$ is quasi-concave in $q$ (Theorem \ref{opt_q_st_theorem}), $w^{'}(q)>0$ holds in the whole range of $q\in(q_m,1]$, which further yields $Q^{'}(q)<0$ according to Lemma \ref{lemma_dQ}.
In other words, $Q(q)$ monotonically decreases in $q$, and thus, $Q(q)>Q(1)>0$, or, $J^{'}(q)>0$.
This means $J(q)$ monotonically increases in $q$, and the optimal $q$ that maximizes $J(q)$ is $q=1$.

2) If $Q(1)\le 0$, combined with $Q(q_m)>0$ in \eqref{F_m}, there at least exists one point $q\in(q_m,1]$ that satisfies $Q(q)=0$ due to the continuity of $Q(q)$ in $q$.
Denote a zero-crossing point $q$ of
$Q(q)$ as $q_{o}$ such that ${Q}(q_o)=0$, or, $J^{'}(q_o)=0$.
To determine the quasi-concavity of $ J(q)$ in $q$, we first take the second-order derivative of $J(q)$ at $q=q_o$, which is
\begin{equation}\label{dJ2}
J^{''}(q_o)={Q^{'}(q_o)}{(1+\rho_cq_o)^{-2}}.
\end{equation}
Recalling $Q(q_o)=0$ yields $w^{'}(q_o)=\frac{Bw(q_o)}{1+\rho_cq_o}>0$.
From Lemma \ref{lemma_dQ}, we obtain $Q^{'}(q_o)<0$, i.e., $J^{''}(q_o)<0$.
This means $J(q)$ is quasi-concave on $q$, and the optimal $q$ that maximizes $J(q)$ is the unique root of equation $Q(q)=0$, i.e., $q=q_o$.
By now, the proof is complete.

\subsection{Proof of Corollary \ref{corollary_q_ee}}
\label{appendix_corollary_q_ee}

Recall \eqref{opt_q_ee_exp}, and the optimal $q^{\circ}_{ee}$ satisfies $Q(q^{\circ}_{ee})=0$.
Similar to the proof of Corollary \ref{corollary_q_st}, we take the first-order derivative of $q^{\circ}_{ee}$ on $w_1(q^{\circ}_{ee})$ and on $w_2(q^{\circ}_{ee})$, respectively, using the derivative rule for implicit functions
with $Q(q^{\circ}_{ee})=0$, and then prove $\frac{d q^{\circ}_{ee}}{d w_1(q^{\circ}_{ee})}<0$ and $\frac{d q^{\circ}_{ee}}{d w_2(q^{\circ}_{ee})}>0$.
Through observing the monotonicity of $w_1(q)$ and $w_2(q)$ with respect to the parameters involved in Corollary \ref{corollary_q_ee}, we can complete the proof.

\subsection{Proof of Corollary \ref{sparse_case}}
\label{appendix_sparse_case}

The expressions of $w_1(q)$ and $w_2(q)$ in \eqref{w} tell that as $\lambda_l\rightarrow 0$, we have  $\frac{w_1(q)-1}{w_1(q)}\rightarrow 1$, $\frac{w_2(q)-1}{w_2(q)}\rightarrow 1$ and $w(q)\rightarrow\frac{\alpha}{2}\ln\Delta$.
Substituting these results into \eqref{opt_q_ee_exp}, we find $Q(q)$ independent of $\lambda_l$ and so is the root $q$ of $Q(q)=0$.
Plugging the obtained solution $q$ into \eqref{ee_max_def}, we can easily conclude that the resulting $\mathbf{\Psi}$ is also independent of $\lambda_l$.
The proof is complete.

\subsection{Proof of Corollary \ref{with_st_constraint}}
\label{appendix_with_st_constraint}
Let us recall \eqref{opt_q_ee}.
Obviously, $q_{ee}^{\star}=\emptyset$ if $q_{ee}^{*}=\emptyset$; $q_{ee}^{\star}=1$ if $q_{ee}^{*}=1$ and $\mathbf{\Omega}(1)>\Omega^{\circ}$ simultaneously hold.
When $q_{ee}^{*}=q_{ee}^{\circ}$, let us distinguish two cases.
In the first case, there is only one root $q\in(q_m,1]$, denoted as $q_{st}^{(1)}$, that satisfies $\mathbf{\Omega}(q)={\Omega}^{\circ}$. If $q_{st}^{(1)}<q_{ee}^{\circ}$, we have $\mathbf{\Omega}(q_{ee}^{\circ})>{\Omega}^{\circ}$ and $q_{ee}^{\star}=q_{ee}^{\circ}$; otherwise, $q_{ee}^{\star}=q_{st}^{(1)}$.
In the second case, there are two roots $q_{st}^{(1)}$ and $q_{st}^{(2)}$ such that $q_{st}^{(1)}<q_{st}^{(2)}$.
In a similar way, we can obtain $q_{ee}^{\star}=q_{ee}^{+}$ with $q_{ee}^{+}$ given in \eqref{opt_q_ee_st}.
By now, the proof is complete.


\begin{thebibliography}{99}
	
	\bibitem{Poor2012Information}
	H. V. Poor, ``Information and inference in the wireless physical layer,'' \emph{IEEE Wireless Commun.}, vol. 19, no. 1, pp. 40--47, Feb. 2012.
	
	\bibitem{Wyner1975Wire-tap}
	A. D. Wyner, ``The wire-tap channel,'' \emph{Bell Syst. Tech. J.}, vol. 54, no. 8, pp. 1355--1387, 1975.
	
	
	\bibitem{Liu2010Multiple}
	R. Liu, T. Liu, H. V. Poor, and S. Shamai, ``Multiple-input multiple-output Gaussian broadcast channels with confidential messages," \emph{IEEE Trans. Inf. Theory}, vol. 56, no. 9, pp. 4215--4227, Sep. 2010.
	
	\bibitem{Zhou2010Secure}
	X. Zhou and M. R. McKay, ``Secure transmission with artificial noise over fading channels: Achievable rate and optimal power allocation,'' \emph{IEEE Trans. Veh. Technol.}, vol. 59, no. 8, pp. 3831--3842, Oct. 2010.
	
	\bibitem{Wang2015Secure}
	H.-M. Wang, T.-X. Zheng, and X.-G. Xia, ``Secure MISO wiretap channels with multiantenna passive eavesdropper: Artificial noise vs. artificial fast fading,'' \emph{IEEE Trans. Wireless Commun.}, vol. 14, no. 1, pp. 94--106, Jan. 2015.
	
	
	\bibitem{Zheng2015Outage}
	T.-X. Zheng, H.-M. Wang, F. Liu, and M. H. Lee, ``Outage constrained secrecy throughput maximization for DF relay networks," \emph{IEEE Trans. Communications}, vol. 63, no. 5, pp. 1741--1755, May 2015.
	
	\bibitem{Zheng2015Multi}
	T.-X. Zheng, H.-M. Wang, J. Yuan, D. Towsley, and M. H. Lee, ``Multi-antenna transmission with artificial noise against randomly distributed eavesdroppers,'' \emph{IEEE Trans. Commun.}, vol. 63, no. 11, pp. 4347--4362, Nov. 2015.
	
	\bibitem{Zheng2016Optimal}
	T.-X. Zheng and H.-M. Wang, ``Optimal power allocation for artificial noise under imperfect CSI against spatially random eavesdroppers,'' \emph{IEEE Trans. Veh. Tech.}, vol. 65, no. 10, pp. 8812--8817, Oct. 2016.
	
	
	\bibitem{Wang2016Physical}
	H.-M. Wang, T.-X. Zheng, J. Yuan, D. Towsley, and M. H. Lee, ``Physical layer security in heterogeneous
	cellular networks," \emph{IEEE Trans. Commun.}, vol. 64, no. 3, pp. 1204--1219, Mar. 2016.
	
	\bibitem{Wang2016Physical Springer}
	H.-M. Wang and T.-X. Zheng, Physical Layer Security in Random Cellular Networks. Singapore: Springer, 2016.

	
	\bibitem{Ma2015Interference}
	C. Ma, J. Liu, X. Tian, H. Yu, Y. Cui, and X. Wang, ``Interference exploitation in D2D-enabled cellular networks: a secrecy
	perspective,'' \emph{IEEE Trans. Commun.}, vol. 63, no. 1, pp. 229--242, Jan. 2015.
	
	\bibitem{Wang Chao2016Physical}
	C. Wang and H.-M. Wang, ``Physical layer security in millimeter wave cellular networks,'' \emph{IEEE Trans. Wireless Commun.}, vol. 15, no. 8, pp. 5569--5585, Aug. 2016.
	

	\bibitem{Zhou2011Throughput}
	X. Zhou, R. Ganti, J. Andrews, and A. Hj{\o}rungnes, ``On the throughput cost of physical tier security in DWNs,'' \emph{IEEE Trans. Wireless Commun.}, vol. 10, no. 8, pp. 2764--2775, Aug. 2011.
	
	\bibitem{Zhang2013Enhancing}
	X. Zhang, X. Zhou, and M. R. McKay, ``Enhancing secrecy with multi-antenna transmission in wireless ad hoc networks,'' \emph{IEEE Trans. Inf. Forensics and Security}, vol. 8, no. 11, pp. 1802--1814, Nov. 2013.
	
	
	\bibitem{Ng2012Energy}
	D. W. K. Ng, E. S. Lo, and R. Schober, ``Energy-efficient resource allocation for
	secure OFDMA systems,'' \emph{IEEE Trans. Veh. Technol.}, vol. 61, no. 6, pp. 2572--2585, July 2012.
	
	\bibitem{Chen2013Energy}
	X. Chen and L. Lei, ``Energy-efficient optimization for physical layer security in multi-antenna downlink networks with QoS guarantee,'' \emph{IEEE Commun. Lett.}, vol. 17, no. 4, pp. 637--640, Apr. 2013.
	
	
	\bibitem{Goel2008Guaranteeing}
	S. Goel and R. Negi, ``Guaranteeing secrecy using artificial noise,'' \emph{IEEE Trans. Wireless Commun.}, vol. 7, no. 6, pp. 2180--2189, Jun. 2008.
	
	\bibitem{Song2016Full}
	L. Song, R. Wichman, Y. Li, and Z. Han, \emph{Full-Duplex Communications and Networks}, Cambridge, UK: Cambridge Univ. Press, in progress, 2016.
	
	\bibitem{Li2012Secure}
	W. Li, M. Ghogho, B. Chen, and C. Xiong, ``Secure communication via
	sending artificial noise by the receiver: Outage secrecy capacity/region
	analysis,'' \emph{IEEE Commun. Lett.}, vol. 16, no. 10, pp. 1628--1631, Oct. 2012.
	
	\bibitem{Zheng2013Improving}
	G. Zheng, I. Krikidis, J. Li, A. Petropulu, and B. Ottersten, ``Improving physical tier secrecy using full-duplex jamming receivers,'' \emph{IEEE Trans.
		Signal Process.}, vol. 61, no. 20, pp. 4962--4974, Oct. 2013.
	
	\bibitem{Zhou2014Application}
	Y. Zhou, Z. Xiang, Y. Zhu, and Z. Xue, ``Application of full-duplex wireless technique
	into secure MIMO communication: achievable
	secrecy rate based optimization," \emph{IEEE Signal Process. Lett.}, vol. 21, no. 7, pp. 804--808, Jul. 2014.
	
	\bibitem{Cepheli2014A_high}
	{\"O}. Cepheli, S. Tedik, and G. K. Kurt, ``A high data rate wireless communication system with improved secrecy: full duplex beamforming," \emph{IEEE Commun. Lett.}, vol. 18, no. 6, pp. 1075--1078, Jun. 2014.
	
	
	\bibitem{Chen2015Physical}
	G. Chen, Y. Gong, P. Xiao,
	and J. A. Chambers, ``Physical layer network security in the
	full-duplex relay system," \emph{IEEE Trans. Inf. Forensics and Security}, vol. 10, no. 3, pp. 574--583, Mar. 2015.
	
	
	\bibitem{Parsaeefard2015Improving}
	S. Parsaeefard, and . Le-Ngoc, ``Improving wireless secrecy rate via full-duplex
	relay-assisted protocols," \emph{IEEE Trans. Inf. Forensics and Security}, vol. 10, no. 10, pp. 2095--2107, Oct. 2015.
	
	
	\bibitem{Zhu2016Physical}
	F. Zhu, F. Gao, T. Zhang, K. Sun, and M. Yao, ``Physical-layer security for full duplex communications with self-interference mitigation," \emph{IEEE Trans. Wireless Commun.}, vol. 15, no. 1, pp. 329--340, Jan. 2016.
	
	\bibitem{Zheng2017Safeguarding}
	T.-X. Zheng, H.-M. Wang, Q. Yang, and M. H. Lee, ``Safeguarding decentralized wireless networks using full-duplex jamming receivers,'' \emph{IEEE Trans. Wireless Commun.}, vol. 16, no. 1, pp. 278--292, Jan. 2017.
	
	
	\bibitem{Haenggi2009Stochastic}
	M. Haenggi, J. Andrews, F. Baccelli, O. Dousse, and M. Franceschetti, ``Stochastic geometry and random graphs for the analysis and design of wireless networks,'' \emph{IEEE J. Sel. Areas Commun.}, vol. 27, no. 7, pp. 1029--1046, Sep. 2009.
	
	\bibitem{Zheng2014Transmission}
	T.-X. Zheng, H.-M. Wang, and Q. Yin, ``On transmission secrecy outage of a multi-antenna system with randomly located eavesdroppers,'' \emph{IEEE Commun. Lett.}, vol. 18, no. 8, pp. 1299--1302, Aug. 2014.

	
	\bibitem{Ha2013Energy}
	D. Ha, K. Lee, and J. Kang, ``Energy efficiency analysis with circuit
	power consumption in massive MIMO systems,'' in \emph{Proc. 2013 Int. Symp. Personal, Indoor Mobile Radio Commun.}, pp. 938--942, London, UK, Sep. 2013.
	
	
	\bibitem{Lee2015Hybrid}
	J. Lee, and T. Q. S. Quek, ``Hybrid full-/half-duplex system analysis in heterogeneous wireless networks,'' \emph{IEEE Trans. Wireless Commun.}, vol. 14, no. 5, pp. 2883--1895, May 2015.
	
	\bibitem{Bharadia2013Full}
	D. Bharadia, E. McMilin, and S. Katti, ``Full duplex radios,'' in \emph{Proc. ACM SIGCOMM 2013}, Hong Kong, China, Aug. 2013.
	
	
	\bibitem{Haenggi2012Stochastic}
	M. Haenggi, \emph{Stochastic Geometry for Wireless Networks}. Cambridge University Press, 2012.
	
	
	\bibitem{Gradshteyn2007Table}
	I. S. Gradshteyn, I. M. Ryzhik, A. Jeffrey, D. Zwillinger, and S. Technica, \emph{Table of Integrals, Series, and Products, 7th ed}. ~New York:
	Academic Press, 2007.
	
	\bibitem{Boyd2004Convex}
	S. Boyd and L. Vandenberghe, \emph{Convex Optimization}. Cambridge, UK: Cambridge Univ. Press, 2004.
	
\end{thebibliography}
\end{document}